\documentclass[manyauthors]{fundam}

\usepackage{graphicx}
\usepackage{amssymb}
\usepackage{amsmath}
\usepackage{braket}
\usepackage{comment}
\usepackage{color}

\newcommand{\id}{\mathrm{id}}
\newcommand{\mat}[1]{\mathbf{Mat}(#1)}
\newcommand{\re}{\mathbf{R}}
\newcommand{\posre}{\mathbf{R}_{\geq 0}}
\newcommand{\complex}{\mathbf{C}}
\newcommand{\rel}{\mathbf{Rel}}

\newcommand{\mon}{\mathbf{Mon}}
\newcommand{\cmon}{\mathbf{C}\mon}
\newcommand{\pcmon}{\mathbf{P}\cmon}
\newcommand{\col}[2]{
 \left[
  \begin{array}{c}
    #1 \\ 
    #2 \\ 
  \end{array}
  \right]}
\newcommand{\row}[2]{
 \left[
  \begin{array}{cc}
    #1 & #2 \\ 
  \end{array}
  \right]}
\newcommand{\colmat}[2]{
 \left(
  \begin{array}{c}
    #1 \\ 
    #2 \\ 
  \end{array}
  \right)}
\newcommand{\rowmat}[2]{
 \left(
  \begin{array}{cc}
    #1 & #2 \\ 
  \end{array}
  \right)}

\newcommand{\rell}[1]{\mathbf{Rel}(#1)}
\usepackage{xy}
\xyoption{all}

\begin{document}

\title{Categorical generalization of spectral decomposition}

\author{
Koki Nishizawa\thanks{This work was supported by JSPS KAKENHI Grant Number JP22K11913.}\corresponding\\
Faculty of Informatics\\
Kanagawa University\\
Japan\\
nishizawa@kanagawa-u.ac.jp
\and
Yusuke Ide\\
Department of Mathematics, \\
College of Humanities and Sciences,\\
Nihon University
\and
Norihiro Tsumagari\\
Center for Education and Innovation\\
Sojo University}

\maketitle

\runninghead{K.Nishizawa, Y.Ide, N.Tsumagari}{Categorical generalization of spectral decomposition}

\begin{abstract}
In this paper, we give several equivalent characterizations for a category with finite biproducts and the sum operation of arrows, and called categories satisfying these semiadditive $\cmon$-categories. 
This allow us to give equivalent structures without directly confirming the existence of biproducts.
Moreover, we define a generalized notion of the spectral decomposition in semiadditive $\cmon$-categories.
We also define the notion of a semiadditive $\cmon$-functor that preserves the spectral decomposition of arrows.
Semiadditive $\cmon$-categories and semiadditive $\cmon$-functors
include many examples.
\end{abstract}

\section{Introduction}

In this paper, we generalize the spectral decomposition of matrices in a category theoretical setting.

The diagonalization involves finding the eigenvalues of a given square matrix on a vector space and decomposing the vector space as the orthogonal direct sum of eigenspaces. 
The original square matrix is represented as the product of a diagonal matrix of the eigenvalues, the change-of-basis matrix consisting of the eigenvectors, and its inverse matrix. 
This representation is called the diagonalization of the original square matrix.
The original square matrix can also be represented as the sum of matrices which project vectors onto each eigenspace and scalar multiple with each eigenvalue. 
This second representation is called the spectral decomposition of the original square matrix.

In this paper, we generalize the spectral decomposition of matrices, characterize the essentially same decomposition in not only vector spaces but also various fields such as graph theory, and provide a basis for defining a transformation method preserving the decomposition between different fields.

For that, we use objects and arrows in categories instead of vector spaces and matrices.
The relationship between the original matrix, eigenvalues, and eigenvectors is generalized as commutative diagrams of arrows.
The situation where the original matrix has a spectral decomposition is also generalized as a condition for arrows.
To do this, we need operations on arrows that generalize matrix sums.
Therefore, we formulate categories equipped with such operations, as the semiadditive $\cmon$-category.

A category with all finite biproducts is called a semiadditive category and
a semiadditive category admits a canonical enrichment over commutative monoids,
by the sum operation of arrows~\cite{MacLane:Working}. 
Conversely, any category enriched over commutative monoids which has either finite products or finite coproducts is semiadditive~\cite{Lack_2012}.
In this paper, we give several equivalent characterizations for a category with finite biproducts and the sum operation of arrows, and call categories satisfying these semiadditive $\cmon$-categories. This allows us to give equivalent structures without directly confirming the existence of biproducts.
We then define what the spectral decomposition of an arrow is.

Moreover, we define a generalized notion of the spectral decomposition in semiadditive $\cmon$-categories.
Our general notion of spectral decomposition includes many examples, including 
separation of disconnected directed graphs,
separation of connected multi-valued directed graphs,
separation of connected directed graphs, and
equitable partition of undirected connected graphs.

We also define the notion of a semiadditive $\cmon$-functor that preserves the spectral decomposition of arrows.
This allows the spectral decomposition in one field to be analyzed in another field.
Our notion of semiadditive $\cmon$-functor includes
the extraction of disconnected 2-valued directed graphs,
from connected 4-valued directed graphs.

This paper is organized as follows. 
Section~\ref{sec:semiad} recalls known definition and examples of semiadditive categories.
In Section~\ref{sec:semiadcmon}, we define the notion of a semiadditive $\cmon$-category.
In Section~\ref{sec:decom}, we generalize spectral decomposition of matrices to the decomposition of arrows in semiadditive $\cmon$-categories.
In Section~\ref{sec:functor}, we define the notion of a semiadditive $\cmon$-functor.
Section~\ref{sec:conclusion} summarizes this work and future work.

\section{Semiadditive categories}\label{sec:semiad}

This section recalls known definition and examples of semiadditive categories~\cite{MacLane:Working}. 

\begin{definition}[zero object]
An object is called a {\bf zero object} (or {\bf null object}), if it is both initial and terminal.
\end{definition}

\begin{definition}[zero morphism]
A family $0$ of {\bf zero morphisms} is defined to be an assignment to each objects $x,y$ of an arrow $0_{x,y}\colon\,x\to y$ satisfying that
\begin{enumerate}
\item for any objects $x,y,z$ and any arrow $h\colon\,x\to y$, $0_{y,z}\circ h=0_{x,z}$ and
\item for any objects $x,y,z$ and any arrow $h\colon\,y\to z$, $h\circ 0_{x,y}=0_{x,z}$.
\end{enumerate}
\end{definition}

If both of $0,0'$ are families of zero morphisms, then $0_{x,y}=0'_{x,y}$, because $0_{x,y}=0'_{y,y}\circ 0_{x,y}=0'_{x,y}$.
Therefore, a family of zero morphisms is unique, if it exists.

\begin{theorem}\label{thm:pointedcat}
Let $c$ be an object of a category $C$. Then the following are equivalent.
\begin{enumerate}
\item $c$ is a zero object of $C$.
\item $C$ has a family $0$ of zero morphisms and $c$ is a terminal object of $C$.
\item $C$ has a family $0$ of zero morphisms and $c$ is an initial object of $C$.
\item $C$ has a family $0$ of zero morphisms and $c$ satisfies $0_{c,c}=\id_{c}$.
\end{enumerate}
\end{theorem}
\begin{proof}
(1. $\Longrightarrow$ 2.) 
By definition of a zero object, $c$ is initial and terminal.
Since $c$ is terminal, there is the unique arrow $\tau_{x}\colon\,x\to c$ for each objects $x$.
Since $c$ is initial, there is the unique arrow $\sigma_{y}\colon\,c\to y$ for each objects $y$.
Let $0_{x,y}$ be $\sigma_{y}\circ\tau_{x}$  for each objects $x,y$.
For any arrow $h\colon\,x\to y$, $0_{y,z}\circ h=\sigma_{z}\circ\tau_{y}\circ h=\sigma_{z}\circ\tau_{x}=0_{x,z}$.
Similarly, for any arrow $h\colon\,y\to z$, $h\circ 0_{x,y}=0_{x,z}$.
Therefore, $0$ is a family of zero morphisms.
(1. $\Longrightarrow$ 3.) 
Similarly, 3 is also implied by 1.
(2. $\Longrightarrow$ 4.) 
By the uniqueness of arrows to terminal object $c$, $0_{c,c}=\id_{c}$.
(3. $\Longrightarrow$ 4.) 
By the uniqueness of arrows from initial object $c$, $0_{c,c}=\id_{c}$.
(4. $\Longrightarrow$ 1.) 
For any object $x$ and arrow $f\colon\,x\to c$, $f=\id_{c,c}\circ f=0_{c,c}\circ f=0_{x,c}$.
Therefore, $c$ is terminal.
Similarly, $c$ is also initial.
\end{proof}

\begin{definition}[pointed category]
A category $C$ is called {\bf pointed}, if $C$ has a zero object $c$ (i.e., when $C$ and $c$ satisfy one (i.e., all) of equivalent conditions of Theorem~\ref{thm:pointedcat}).
\end{definition}

An object $c_{1}\times c_{2}$ is called a {\bf binary product} of $c_{1},c_{2}$ and arrows $\pi_{1}^{c_{1},c_{2}}\colon\,c_{1}\times c_{2}\to c_{1}$ and 
$\pi_{2}^{c_{1},c_{2}}\colon\,c_{1}\times c_{2}\to c_{2}$ are called its {\bf projections}, if for each object $x$, each arrow $f_{1}\colon\,x\to c_{1}$, each arrow $f_{2}\colon\,x\to c_{2}$,
there exists a unique arrow $\col{f_{1}}{f_{2}}\colon\,x\to c_{1}\times c_{2}$ satisfying $\pi_{1}^{c_{1},c_{2}}\circ\col{f_{1}}{f_{2}}=f_{1}$ and $\pi_{2}^{c_{1},c_{2}}\circ \col{f_{1}}{f_{2}}=f_{2}$.

\begin{figure}[h]
 \begin{center}
  \includegraphics[bb=0 0 442 259,width=0.4\columnwidth]{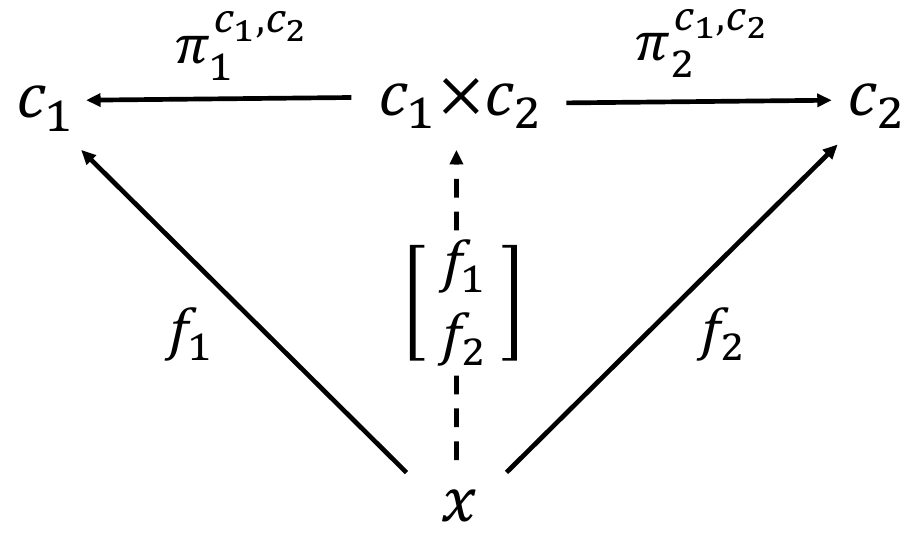}
 \end{center}
\end{figure}

The uniqueness implies $\col{f_{1}}{f_{2}}\circ g=\col{f_{1}\circ g}{f_{2}\circ g}$.
The arrow $\col{\id_{x}}{\id_{x}}\colon\,x\to x\times x$ is also denoted $\Delta_{x}$.

An object $c_{1}+ c_{2}$ is called a {\bf binary coproduct} of $c_{1},c_{2}$ and arrows $\iota_{1}^{c_{1},c_{2}}\colon\,c_{1}\to c_{1}+c_{2}$ and 
$\iota_{2}^{c_{1},c_{2}}\colon\,c_{2}\to c_{1}+c_{2}$ are called its {\bf injections}, if for each object $x$, each arrow $f_{1}\colon\,c_{1}\to x$, each arrow $f_{2}\colon\,c_{2}\to x$,
there exists a unique arrow $\row{f_{1}}{f_{2}}\colon\,c_{1}+c_{2}\to x$ satisfying $\row{f_{1}}{f_{2}}\circ\iota_{1}^{c_{1},c_{2}}=f_{1}$ and $\row{f_{1}}{f_{2}}\circ\iota_{2}^{c_{1},c_{2}}=f_{2}$. 

\begin{figure}[h]
 \begin{center}
  \includegraphics[bb=0 0 442 283,width=0.4\columnwidth]{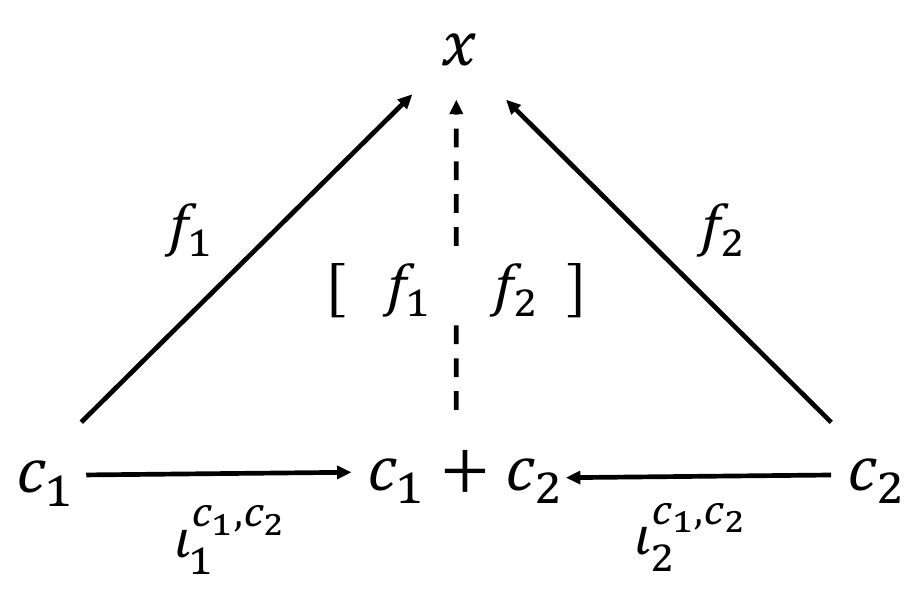}
 \end{center}
\end{figure}

The uniqueness implies $g\circ\row{f_{1}}{f_{2}}=\row{g\circ f_{1}}{g\circ f_{2}}$.
The arrow $\row{\id_{x}}{\id_{x}}\colon\,x+x\to x$ is also denoted $\nabla_{x}$.

For any $f_{1,1}\colon\,x_{1}\to c_{1}$, $f_{1,2}\colon\,x_{1}\to c_{2}$, $f_{2,1}\colon\,x_{2}\to c_{1}$, $f_{2,2}\colon\,x_{2}\to c_{2}$, 
the arrow $\col{\row{f_{1,1}}{f_{2,1}}}{\row{f_{1,2}}{f_{2,2}}}$
$\colon\,x_{1}+x_{2}\to c_{1}\times c_{2}$ equals 
 $\row{\col{f_{1,1}}{f_{1,2}}}{\col{f_{2,1}}{f_{2,2}}}$.
 So, it is also denoted
 $\left[
  \begin{array}{cc}
    f_{1,1} & f_{2,1} \\ 
    f_{1,2} & f_{2,2} \\ 
  \end{array}
  \right]$.

\begin{definition}[biproduct]
Let $C$ be a pointed category.
Let $c_{1}\times c_{2}$ be a binary product of $c_{1},c_{2}$ with projections $\pi_{1}^{c_{1},c_{2}},\pi_{2}^{c_{1},c_{2}}$.
Let $c_{1}+c_{2}$ be a binary coproduct of $c_{1},c_{2}$ with injections $\iota_{1}^{c_{1},c_{2}},\iota_{2}^{c_{1},c_{2}}$.
If the arrow 
 $\left[
  \begin{array}{cc}
    \id_{c_{1}} & 0_{c_{2},c_{1}} \\ 
    0_{c_{1},c_{2}} & \id_{c_{2}} \\ 
  \end{array}
  \right]$ is an isomorphism, then the isomorphic objects $c_{1}\times c_{2}$ and $c_{1}+c_{2}$ are
called the {\bf binary biproduct} of $c_{1},c_{2}$ and denoted $c_{1}\oplus c_{2}$.  
\end{definition}

We can write $\pi_{1},\pi_{2}$ for $\pi_{1}^{c_{1},c_{2}},\pi_{2}^{c_{1},c_{2}}$ when $c_{1},c_{2}$ are clear.

The biproduct $c_{1}\oplus c_{2}$ is both product and coproduct of $c_{1},c_{2}$ with projections and injections satisfying
$\pi_{1}=\row{\id_{c_{1}}}{0_{c_{2},c_{1}}}$, 
$\pi_{2}=\row{0_{c_{1},c_{2}}}{\id_{c_{2}}}$, 
$\iota_{1}=\col{\id_{c_{1}}}{0_{c_{1},c_{2}}}$, and 
$\iota_{2}=\col{0_{c_{2},c_{1}}}{\id_{c_{2}}}$.
The composition of arrows through biproducts satisfies 
$\row{f}{0}\circ\col{h}{k}=f\circ h=\row{f}{g}\circ\col{h}{0}$ and
$\row{0}{g}\circ\col{h}{k}=g\circ k=\row{f}{g}\circ\col{0}{k}$.

For any arrows $f,g\colon\,x\to y$, the arrow 
 $\left[
  \begin{array}{cc}
    f & 0_{x,y} \\ 
    0_{x,y} & g\\ 
  \end{array}
  \right]$ 
is also denoted $f\oplus g\colon\,x\oplus x\to y\oplus y$.

Similarly, biproducts (called {\bf finite biproducts}) of finite sequences $c_{1},c_{2},\ldots,c_{m}$ of objects are defined and denoted $\bigoplus^{m}_{i=1} c_{i}$.  
Particularly, the biproduct of the empty sequence is defined to be a zero object.
A category $C$ has all finite biproducts, if and only if $C$ has a zero object and all binary biproducts.

\begin{definition}[semiadditive category]
A {\bf semiadditive} category is defined to be a tuple $(C,0,\oplus,\pi,\iota)$ such that
\begin{enumerate}
\item $C$ is a category,
\item $0$ is a family of zero morphisms,
\item for each pair of objects $c_{1},c_{2}$ of $C$, $c_{1}\oplus c_{2}$ is a binary biproduct with projections $\pi_{1}^{c_{1},c_{2}},\pi_{2}^{c_{1},c_{2}}$ and injections $\iota_{1}^{c_{1},c_{2}},\iota_{2}^{c_{1},c_{2}}$,
\item $\pi$ is a family of pairs of projections $\pi_{1}^{c_{1},c_{2}},\pi_{2}^{c_{1},c_{2}}$, and
\item $\iota$ is a family of pairs of injections $\iota_{1}^{c_{1},c_{2}},\iota_{2}^{c_{1},c_{2}}$.
\end{enumerate}
\end{definition}

\begin{lemma}\label{lem:innerproduct}
For any semiadditive category $(C,0,\oplus,\pi,\iota)$,
for any arrows $h\colon\,w\to x$, $k\colon\,w\to y$, $f\colon\,x\to z$, $g\colon\,y\to z$ in $C$, 
$\row{f}{g}\circ\col{h}{k}$ equals $\nabla_{z}\circ((f\circ h)\oplus (g\circ k))\circ\Delta_{w}$.
\end{lemma}
\begin{proof}
\[
\begin{array}{ll}
   & \row{f}{g}\circ\col{h}{k}\\
= & \row{
         \row{
            \id_{z}
         }{
            \id_{z}
         }
         \circ
         \col{
            f
         }{
            0_{x,z}
         }
      }{
         \row{
            \id_{z}
         }{
            \id_{z}
         }
         \circ
         \col{
            0_{y,z}
         }{
            g
         }
      }
      \circ
      \col{
         \row{
            h
         }{
            0_{w,x}
         }
         \circ
         \col{
            \id_{w}
         }{
            \id_{w}
         }
      }{
         \row{
            0_{w,y}
         }{
            k
         }
         \circ
         \col{
            \id_{w}
         }{
            \id_{w}
         }
      }\\
= & \row{\id_{z}}{\id_{z}}\circ
 \left[
  \begin{array}{cc}
    f & 0_{y,z} \\ 
    0_{x,z} & g\\ 
  \end{array}
  \right]\circ
 \left[
  \begin{array}{cc}
    h & 0_{w,x} \\ 
    0_{w,y} & k\\ 
  \end{array}
  \right]
\circ\col{\id_{w}}{\id_{w}}\\
= & \row{\id_{z}}{\id_{z}}\circ
\row{
 \left[
  \begin{array}{cc}
    f & 0_{y,z} \\ 
    0_{x,z} & g\\ 
  \end{array}
  \right]\circ\col{h}{0_{w,y}}
}{
 \left[
  \begin{array}{cc}
    f & 0_{y,z} \\ 
    0_{x,z} & g\\ 
  \end{array}
  \right]\circ\col{0_{w,x}}{k}
}
 \circ\col{\id_{w}}{\id_{w}}\\
= & \row{\id_{z}}{\id_{z}}\circ
 \left[
  \begin{array}{cc}
    f\circ h & 0_{w,z} \\ 
    0_{w,z} & g\circ k\\ 
  \end{array}
  \right]
\circ\col{\id_{w}}{\id_{w}}\\
= & \nabla_{z}\circ((f\circ h)\oplus (g\circ k))\circ\Delta_{w}\\
\end{array}
\]  
\end{proof}

The unique arrow $\col{f}{g}$ depends on the choice of a binary product and projections.
Since $\Delta$ is $\col{\id}{\id}$, it also depends on these.
Similarly, $\row{f}{g}$ and $\nabla$ depends on the choice of a binary coproduct and injections.

\begin{example}
We write $\mat{\re}$ for the category defined by following data.
\begin{itemize}
\item Objects are non-negative integers.
\item Arrows from $m$ to $n$ are $n\times m$ real matrices.
\item For $g\colon m \to l$ and $f\colon  n\to m$, the composition $g\circ f\colon  n\to l$ is defined by the matrix multiplication $g\times f$.
\item An identity ${\rm id}_n\colon n\to n$ is the identity matrix $I_n$ of size $n$. 
\end{itemize}
Then, the tuple $(\mat{\re},0,\oplus,\pi,\iota)$ of the following data is a semiadditive category.
\begin{itemize}
\item The zero object is the integer zero. 
\item $0$ is the family of zero matricies as zero morphisms. 
\item For any objects $m$ and $n$, the biproduct $m\oplus n$ of $m$ and $n$ is defined by the sum $m+n$
with projections
$\rowmat{f}{0_{n,m}},\rowmat{0_{m,n}}{g}$ as $\pi_{1},\pi_{2}$
and injections
$\colmat{f^{-1}}{0_{m,n}},\colmat{0_{n,m}}{g^{-1}}$ as $\iota_{1},\iota_{2}$, 
where $f$ is an $m$-dimensional regular matrix, $g$ is an $n$-dimensional regular matrix, and $f^{-1}$ and $g^{-1}$ are their inverses. Here, the round brackets $()$ represents usual matricies.
\end{itemize}
{
In this case, the arrow $\col{f_{1}}{f_{2}}$ is defined by $\col{f_{1}}{f_{2}} = \colmat{f^{-1}f_{1}}{g^{-1}f_{2}}$.
Therefore, 
$
\pi _{1} \col{f_{1}}{f_{2}}
=
\rowmat{f}{0_{n,m}}  \colmat{f^{-1}f_{1}}{g^{-1}f_{2}}
=
f_{1} 
$
and 
$
\pi _{2} \col{f_{1}}{f_{2}}
=
\colmat{0_{n,m}}{g^{-1}}  \colmat{f^{-1}f_{1}}{g^{-1}f_{2}}
=
f_{2} 
$ hold. 
Also the arrow $\row{f_{1}}{f_{2}}$ is defined by $\row{f_{1}}{f_{2}} = \rowmat{f_{1}f}{f_{2}g}$. 
}
\end{example}

\begin{example}
We write $\mat{\posre}$ for the category defined by following data.
\begin{itemize}
\item Objects are non-negative integers.
\item Arrows from $m$ to $n$ are $n\times m$ non-negative real matrices. 
\item The composition of arrows and an identity arrow are defined as same as $\mat{\re}$. 
\end{itemize}
Then, the tuple $(\mat{\posre},0,\oplus,\pi,\iota)$ of the following data is a semiadditive category.
\begin{itemize}
\item The zero object, the family $0$ of zero morphisms, and the biproduct $m\oplus n$ of $m$ and $n$ are also defined similarly to $\mat{\re}$. 
\item The projections $\pi_{1},\pi_{2}$ are given by $\rowmat{f}{0_{n,m}},\rowmat{0_{m,n}}{g}$ and injections $\iota_1,\, \iota_2$ are given by $\colmat{f^{-1}}{0_{m,n}},\colmat{0_{n,m}}{g^{-1}}$, where $f,g$ are non-negative monomial matrices, because the inverses of monomial matrices are also non-negative.
\end{itemize}
\end{example}

\begin{example}
We write $\mat{\complex}$ for the category defined by following data.
\begin{itemize}
\item Objects are non-negative integers.
\item Arrows from $m$ to $n$ are $n\times m$ complex matrices.
\end{itemize}
Then, the semiadditive category $(\mat{\complex},0,\oplus,\pi,\iota)$ is defined similar to $\mat{\re}$.
\end{example}

\begin{example}\label{example-rel}
We write $\rel$ for the category defined by following data.
\begin{itemize}
\item Objects are sets.
\item Arrows from $A$ to $B$ are binary relations (i.e., subsets of $B\times A$).
\item For $g\colon B \to C$ and $f\colon  A\to B$, the composition $g\circ f\colon A\to C$ is defined by the composition of relations, i.e.
\[
g\circ f=\{(c, a)\mid \exists b\in B, (c,b)\in g, (b,a)\in f\}
\]
\item An identity $\id_{X}\colon\,X\to X$ is the identity relation $\{(a,a)\mid a\in X\}$. 
\end{itemize}

Then, the tuple $(\rel,0,\oplus,\pi,\iota)$ of the following data is a semiadditive category.
\begin{itemize}
\item The zero object is the empty set $\emptyset$. 
\item $0$ is the family of empty relations. 
\item For any objects $A$ and $B$, the biproduct $A\oplus B$ of $A$ and $B$ is defined by, for example, 
$$A\oplus B=\{(1,a)\mid a\in A\}\cup\{(2,b)\mid b\in B\}.$$ 
Let $f\colon A\to A$ and $g\colon B\to B$ be bijective mappings, that is, satisfying $f\circ f^\top=f^\top \circ f= \id_{A}$ and $g\circ g^\top=g^\top \circ g= \id_{B}$ where $(-)^\top$ denotes the converse of its relation. 
The projections $\pi_1, \pi_2$ are given by 
$$\pi_{1}= f\circ p_1\mbox{ and } \pi_{2}= g\circ p_2, $$ 
where 
$p_{1}: A\oplus B \to A$ and $p_{2}: A\oplus B \to B$ are defined by  
$$p_{1}= \{(a,(1,a))\mid a\in A\} \mbox{ and } 
p_{2}= \{(b,(2,b))\mid b\in B\}.$$
The injections are given by 
$\iota_{1}= \pi_{1}^\top=p_1^\top \circ f^\top$ and  
$\iota_{2}= \pi_{2}^\top=p_2^\top\circ g^\top$. \\[5pt]
For any $h_1\colon X\to A$ and $h_2\colon X\to B$, the arrows $\col{h_1}{h_2}$ \vspace{5pt} and $\row{h_1}{h_2}$ are given by 
$
\col{h_1}{h_2}=
\iota_1^{A,B} \circ h_1\, \cup \,\iota_2^{A,B}\circ h_2
$ 
and \vspace{5pt}
$
\row{h_1}{h_2}=
h_1\circ \pi_1^{A,B} \,\cup\,h_2\circ \pi_2^{A,B}
$, 
respectively. 
Then for any 
$h_{11}\colon  X_1 \to A$, $h_{12}\colon  X_1 \to B$, $h_{21}\colon  X_2 \to A$, $h_{22}\colon  X_2 \to B$, 
the following equation holds:
$$
\begin{array}{rcl}
 \left[
  \begin{array}{cc}
    h_{11} & h_{12} \\ 
    h_{21} & h_{22}\\ 
  \end{array}
  \right]
&=&
\iota_1^{A,B}\circ h_{11}\circ \pi_1^{A,B}\,
\cup \, \iota_2^{A,B}\circ h_{12}\circ \pi_1^{A,B}\,\\[5pt]
&&\quad \cup \, \iota_1^{A,B}\circ h_{21}\circ \pi_2^{A,B}\,
\cup \, \iota_2^{A,B}\circ h_{22}\circ \pi_2^{A,B}\,.
\end{array}
$$
Note that 
$\Delta_X =\col{\id_X}{\id_X}=\iota_1^{X,X}\cup \iota_2^{X,X}$\vspace{3pt} and 
$\nabla_X =\row{\id_X}{\id_X}=\pi_1^{X,X}\cup \pi_2^{X,X}$ 
for each set $X$ on $\rel$.
\end{itemize}
\end{example}

The next example shows a generalization of $\rel$. 

A {\it Heyting algebra} $(L, \wedge, \vee, \Rightarrow)$ is a lattice $(L, \wedge, \vee)$ equipped with the implication operation $\Rightarrow$ satisfying the property that 
$(x\wedge a)\leq b$ if and only if $x\leq (a\Rightarrow b)$. 
If $L$ is a complete lattice, a Heyting algebra $L$ is called a {\it complete Heyting algebra}. 
A complete Heyting algebra $L$ has the least element $\bot$ and the greatest element $\top$ that satisfies $(a\Rightarrow a)=\top$ for any element $a\in L$. 

An example of complete Heyting algebras is the unit interval $[0,1]$ 
where ${a\wedge b}= {\rm min}\{a,b\}$, $a\vee b= {\rm max}\{a,b\}$ and the implication $\Rightarrow$ satisfies that ${(a\Rightarrow b)}=1$ if $a\leq b$, and $(a\Rightarrow b)=b$ if $a>b$. 

For use in the later discussion, we introduce the 4-element Boolean algebra $B_4=\{0,a,b,1\}$ as an instance of complete Heyting algebras. It is satisfied that $0<a<1$, $0<b<1$, $a\wedge b=0$, $a\vee b=1$, and $a$ and $b$ are incomparable. Then $B_4$ is a complete Heyting algebra where $\bot=0$, $\top=1$, and $(x\Rightarrow y) = \neg x \vee y$. 

\begin{figure}[h]
 \begin{center}
  \includegraphics[bb=0 0 245 247,width=0.2\columnwidth]{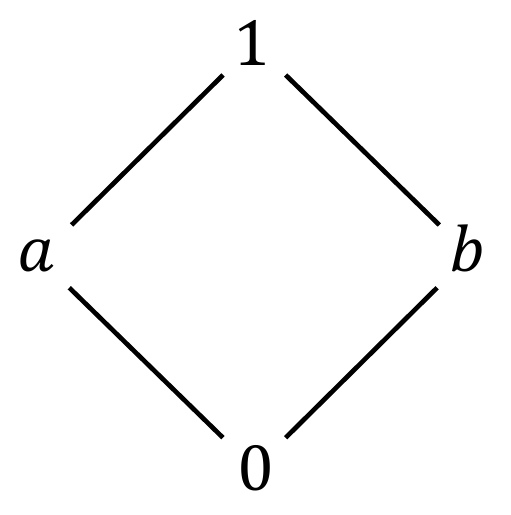}
 \end{center}
\end{figure}

\begin{example}\label{example-Lrel}
Let $L$ be a complete Heyting algebra. We write $\rell{L}$ for the category defined by following data.
\begin{itemize}
\item Objects are sets.
\item Arrows from $A$ to $B$ are $L$-relations, that is, mappings from $B\times A$ to $L$.
\item For $g\colon B \to C$ and $f\colon  A\to B$, the composition $g\circ f\colon A\to C$ is defined by 
\[
(g\circ f) (c,a) =\underset{b\in B}{\bigvee}  \{g(c,b)\wedge f(b,a)\}
\]
\item An identity $\id_{X}\colon\,X\to X$ is defined as $\id_X(b,a) =\top$ if $a=b$, otherwise $\id_X(b,a) =\bot$. 
\end{itemize}
After this, $h^\top$ denotes the converse of a $L$-relation $h\colon A\to B$, that is, 
the arrow $h^\top \colon B\to A$ satisfies $h^\top (a,b) = h(b, a)$ for any $a\in A$ and $b\in B$. 
And $h_1 \sqcup h_2$ denotes the join of two arrows $h_1, h_2 \colon A\to B$, defined by $(h_1 \sqcup h_2)(b, a) = h_1(b, a) \vee h_2(b, a) $ for any $a\in A$ and $b\in B$. 
Because of the distributivities on $L$, it holds that 
$(h_1 \sqcup h_2)\circ f=h_1\circ f \sqcup h_2\circ f$ and 
$g\circ (h_1 \sqcup h_2)=g\circ h_1 \sqcup g\circ h_2$ 
for any  $f\colon X\to A$ and $g\colon B\to Y$.

The tuple $(\rell{L},0,\oplus,\pi,\iota)$ of the following data is a semiadditive category.
\begin{itemize}
\item The zero object is the empty set $\emptyset$. 
\item $0$ is the family of the $L$-relations mapping only to $\bot$.
\item 
For any objects $A$ and $B$, the biproduct $A\oplus B$ of $A$ and $B$ is defined, for example, as same as $\rel$. 
Let $f:A\to A$ and $g:B\to B$ be isomorphisms, that is, satisfying that $f\circ f^\top=f^\top \circ f= \id_{A}$ and $g\circ g^\top=g^\top \circ g= \id_{B}$.  
The projections $\pi_1, \pi_2$ are given by 
$\pi_{1}= f\circ p_1\mbox{ and } \pi_{2}= g\circ p_2, $
where 
$p_{1}: A\oplus B \to A$ and $p_{2}: A\oplus B \to B$ are defined by 
$$
p_{1}(a,(i,x))= \left\{
\begin{array}{ll}
\top & \quad \mbox{if}\,  (i,x)=(1,a)\\
\bot & \quad \mbox{otherwise}
\end{array}
\right.
\quad 
p_{2}(b,(i,x))= \left\{
\begin{array}{ll}
\top & \quad \mbox{if}\,  (i,x)=(2,b)\\
\bot & \quad \mbox{otherwise}
\end{array}
\right.
$$
for each $a\in A$, $b\in B$ and $(i,x)\in A\oplus B$. \\[5pt]
The injections are given by 
$\iota_{1}= \pi_{1}^\top=p_1^\top \circ f^\top$ and  
$\iota_{2}= \pi_{2}^\top=p_2^\top\circ g^\top$. \\[5pt]
For any $h_1:X\to A$ and $h_2:X\to B$, the arrows $\col{h_1}{h_2}$ \vspace{3pt} and $\row{h_1}{h_2}$ are given by 
$
\col{h_1}{h_2}=
\iota_1^{A,B} \circ h_1\, \sqcup \,\iota_2^{A,B}\circ h_2
$ 
and \vspace{5pt}
$
\row{h_1}{h_2}=
h_1\circ \pi_1^{A,B} \,\sqcup\,h_2\circ \pi_2^{A,B}
$.
Note that 
$\Delta_X =\col{\id_X}{\id_X}=\iota_1^{X,X}\sqcup \iota_2^{X,X}$\vspace{3pt} and
$\nabla_X =\row{\id_X}{\id_X}=\pi_1^{X,X}\sqcup \pi_2^{X,X}$ 
for each set $X$ on $\rell{L}$.
\end{itemize}

\end{example}

Using the unit interval $[0,1]$ as a complete Heyting algebra $L$, $[0,1]$-relations are called fuzzy relations. 
Obviously $\rell{L}$ is an expansion of $\rel$. 
In fact, if $L$ is a set $\{0,1\}$ then $\{0,1\}$-relations are equivalent to binary relations. 
Therefore $\rel \simeq \rell{\{0,1\}}$. 

The category $\rell{B_4}$ is also an example of the category of $L$-relations.

\section{Characterization of semiadditive $\cmon$-categories}\label{sec:semiadcmon}

In this paper, we give several equivalent characterizations for a category with finite biproducts and the sum operation of arrows, and call categories satisfying these semiadditive $\cmon$-categories. 

\begin{definition}[$\mon$-category]
A {\bf $\mon$-category} is defined to be a tuple $(C,+,0)$ satisfying the following conditions.
\begin{enumerate}
\item $C$ is a pointed category and $0$ is the family of zero morphisms.
\item $+$ is an assignment to each objects $x,y$ of a binary operator $+_{x,y}$ on the homset $C(x,y)$ from $x$ to $y$.
\item For each objects $x,y$, $+_{x,y}$ is associative and the zero morphism $0_{x,y}$ is its unit.
\item Composition of arrows preserves $+$.
\begin{itemize}
\item For $f,g\colon\,y\to z$ and $h\colon\,x\to y$, $(f+_{y,z}g)\circ h=f\circ h+_{x,z}g\circ h$
\item For $f,g\colon\,x\to y$ and $h\colon\,y\to z$, $h\circ(f+_{x,y}g)=h\circ f+_{x,z}h\circ g$
\end{itemize}
\end{enumerate}
If every $+_{x,y}$ is commutative, then $(C,+,0)$ is called a {\bf $\cmon$-category}.
\end{definition}

A $\mon$-category and a $\cmon$-category can be regarded as a category enriched by 
the monoidal category $\mon$ of monoids and the monoidal category $\cmon$ of commutative monoids~\cite{Kelly1982}, respectively.

\begin{theorem}\label{thm:semiadditive}
Let $C$ be a pointed category and $0$ be the family of zero morphisms.
For each pair of objects $x,y$ of $C$, let $+_{x,y}$ be an binary operator on the homset $C(x,y)$.
For each pair of objects $x,y$ of $C$, let $x\oplus y$ be an object of $C$.
Let $\pi$ be a family of pairs of arrows $\pi_{1}^{x_{1},x_{2}}\colon\,x_{1}\oplus x_{2}\to x_{1},\pi_{2}^{x_{1},x_{2}}\colon\,x_{1}\oplus x_{2}\to x_{2}$.
Let $\iota$ be a family of pairs of arrows $\iota_{1}^{x_{1},x_{2}}\colon\,x_{1}\to x_{1}\oplus x_{2},\iota_{2}^{x_{1},x_{2}}\colon\,x_{2}\to x_{1}\oplus x_{2}$.
Then the following are equivalent.
\begin{enumerate}
\item $(C,0,\oplus,\pi,\iota)$ is semiadditive and for each arrows $f,g\colon\,x\to y$, $f+_{x,y}g=\nabla_{y}\circ(f\oplus g)\circ\Delta_{x}$.
\item $(C,+,0)$ is a $\cmon$-category and $(C,0,\oplus,\pi,\iota)$ is semiadditive.
\item $(C,+,0)$ is a $\mon$-category. For each pair of objects $x_{1},x_{2}$ of $C$,
  $x_{1}\oplus x_{2}$ is the binary product of $x_{1},x_{2}$ with projections $\pi_{1}^{x_{1},x_{2}},\pi_{2}^{x_{1},x_{2}}$.
  For the product $x_{1}\oplus x_{2}$, $\iota_{1}^{x_{1},x_{2}}=\col{\id_{x_{1}}}{0_{x_{1},x_{2}}}$ and $\iota_{2}^{x_{1},x_{2}}=\col{0_{x_{2},x_{1}}}{\id_{x_{2}}}$.
\item $(C,+,0)$ is a $\mon$-category. For each pair of objects $x_{1},x_{2}$ of $C$,
  $x_{1}\oplus x_{2}$ is the binary coproduct of $x_{1},x_{2}$ with injections $\iota_{1}^{x_{1},x_{2}},\iota_{2}^{x_{1},x_{2}}$.
  For the coproduct $x_{1}\oplus x_{2}$, $\pi_{1}^{x_{1},x_{2}}=\row{\id_{x_{1}}}{0_{x_{2},x_{1}}}$ and $\pi_{2}^{x_{1},x_{2}}=\row{0_{x_{1},x_{2}}}{\id_{x_{2}}}$.
\item $(C,+,0)$ is a $\mon$-category and for each pair of objects $x_{1},x_{2}$ of $C$, the conditions (a)-(e) hold.
\begin{enumerate}
  \item $\pi_{1}^{x_{1},x_{2}}\circ\iota_{1}^{x_{1},x_{2}}=\id_{x_{1}}$,
  \item $\pi_{2}^{x_{1},x_{2}}\circ\iota_{2}^{x_{1},x_{2}}=\id_{x_{2}}$,
  \item $\pi_{1}^{x_{1},x_{2}}\circ\iota_{2}^{x_{1},x_{2}}=0_{x_{2},x_{1}}$,
  \item $\pi_{2}^{x_{1},x_{2}}\circ\iota_{1}^{x_{1},x_{2}}=0_{x_{1},x_{2}}$, and
  \item $\iota_{1}^{x_{1},x_{2}}\circ\pi_{1}^{x_{1},x_{2}}+_{x_{1}\oplus x_{2}, x_{1}\oplus x_{2}}\iota_{2}^{x_{1},x_{2}}\circ\pi_{2}^{x_{1},x_{2}}=\id_{x_{1}\oplus x_{2}}$.
 \end{enumerate}
\end{enumerate}
\end{theorem}
\begin{proof}
(1. $\Longrightarrow$ 2.)
Assume that $(C,0,\oplus,\pi,\iota)$ is semiadditive and $f+_{x,y}g=\nabla_{y}\circ(f\oplus g)\circ\Delta_{x}$.
The binary operator $+_{x,y}$ on $C(x,y)$ is associative, as follows. 
\[
\begin{array}{ll}
 & (f+_{x,y}g)+_{x,y}h\\
= & \nabla_{y}\circ((f+_{x,y}g)\oplus h)\circ\Delta_{x}\\  
= & \nabla_{y}\circ((\nabla_{y}\circ(f\oplus g)\circ\Delta_{x})\oplus h)\circ\Delta_{x}\\  
= & \nabla_{y}\circ(\nabla_{y} \oplus \id_{y})\circ((f \oplus g) \oplus h)\circ(\Delta_{x} \oplus \id_{x})\circ\Delta_{x}\\
= & \row{\row{f}{g}}{h}\circ(\Delta_{x} \oplus \id_{x})\circ\Delta_{x}\\
= & \nabla_{y}\circ(\id_{y}\oplus \nabla_{y})
  \circ
 \left[
  \begin{array}{cc}
    \row{f}{0_{x,y}} & 0_{x,y} \\ 
 \left[
  \begin{array}{cc}
    0_{x,y} & g \\ 
    0_{x,y} & 0_{x,y}\\ 
  \end{array}
  \right]   
    & \col{0_{x,y}}{h}\\ 
  \end{array}
  \right]  
    \circ
    (\Delta_{x} \oplus \id_{x})\circ\Delta_{x}\\
= & \nabla_{y}\circ(\id_{y}\oplus \nabla_{y})\circ\col{f}{\col{g}{h}}\\
= & \nabla_{y}\circ(\id_{y}\oplus \nabla_{y})\circ(f\oplus(g\oplus h))\circ(\id_{x} \oplus \Delta_{x})\circ\Delta_{x}\\
= & \nabla_{y}\circ(f\oplus (\nabla_{y}\circ(g\oplus h)\circ\Delta_{x}))\circ\Delta_{x}\\  
= & \nabla_{y}\circ(f\oplus (g+_{x,y}h))\circ\Delta_{x}\\  
= & f+_{x,y}(g+_{x,y}h)\\
 \end{array}
\]  

The binary operator $+_{x,y}$ is commutative, as follows.
\[
\begin{array}{ll}
 & f+_{x,y}g\\
= & \nabla_{y}\circ(f\oplus g)\circ\Delta_{x}\\  
= & \row{f}{g}\circ\Delta_{x}\\
= & \nabla_{y}\circ
 \left[
  \begin{array}{cc}
    0_{x,y} & g\\ 
    f & 0_{x,y} \\ 
  \end{array}
  \right]
\circ\Delta_{x}\\
= & \nabla_{y}\circ\col{g}{f}\\
= & \nabla_{y}\circ(g\oplus f)\circ\Delta_{x}\\  
= & g+_{x,y}f\\
\end{array}
\]  

The zero morphism $0_{x,y}$ is the unit of $+_{x,y}$, because for each arrow $f\colon\,x\to y$, 
$f+_{x,y}0_{x,y}=\nabla_{y}\circ(f\oplus 0_{x,y})\circ\Delta_{x}=f$ and $0_{x,y}+_{x,y}f=f$ similarly. 

Composition of arrows preserves $+$, because $(f+_{x,y}g)\circ h=(f\circ h)+_{w,y}(g\circ h)$ for each arrows $f,g\colon\,x\to y, h\colon\,w\to x$  as follows.
\[
\begin{array}{ll}
 & (f+_{x,y}g)\circ h\\
= & \nabla_{y}\circ(f\oplus g)\circ\Delta_{x}\circ h\\  
= & \nabla_{y}\circ\col{f\circ h}{g\circ h}\\
= & \nabla_{y}\circ((f\circ h)\oplus(g\circ h))\circ\Delta_{w}\\  
= & (f\circ h)+_{w,y}(g\circ h)\\
\end{array}
\]  
Similarly, $h\circ(f+_{x,y}g)=(h\circ f)+_{x,z}(h\circ g)$ for each arrows $f,g\colon\,x\to y, h\colon\,y\to z$.
Therefore, $C$ is a $\cmon$-category.

(2. $\Longrightarrow$ 3.) 
By definition of a biproduct.

(2. $\Longrightarrow$ 4.) 
By definition of a biproduct.

(3. $\Longrightarrow$ 5.) 
The conditions (a),(d) are implied by $\iota_{1}^{x_{1},x_{2}}=\col{\id_{x_{1}}}{0_{x_{1},x_{2}}}$ and 
(b),(c) are implied by $\iota_{2}^{x_{1},x_{2}}=\col{0_{x_{2},x_{1}}}{\id_{x_{2}}}$.
By the definition of products, there is the unique arrow $k\colon\,x_{1}\oplus x_{2}\to x_{1}\oplus x_{2}$ satisfying 
$\pi_{1}\circ k=\pi_{1}$ and $\pi_{2}\circ k=\pi_{2}$.
Here, $k=\id_{x_{1}\oplus x_{2}}$ satisfies that and $k=\iota_{1}\circ\pi_{1}+\iota_{2}\circ\pi_{2}$ also satisfies as follows.
\[
\pi_{1}\circ k=\pi_{1}\circ (\iota_{1}\circ\pi_{1}+\iota_{2}\circ\pi_{2})=\pi_{1}\circ\iota_{1}\circ\pi_{1}+\pi_{1}\circ\iota_{2}\circ\pi_{2}=\id_{x_{1}}\circ\pi_{1}+0_{x_{2},x_{1}}\circ\pi_{2}=\pi_{1}
\]
\[
\pi_{2}\circ k=\pi_{2}\circ (\iota_{1}\circ\pi_{1}+\iota_{2}\circ\pi_{2})=\pi_{2}\circ\iota_{1}\circ\pi_{1}+\pi_{2}\circ\iota_{2}\circ\pi_{2}=0_{x_{1},x_{2}}\circ\pi_{1}+\id_{x_{2}}\circ\pi_{2}=\pi_{2}
\]
Therefore, we have the condition (e) $\iota_{1}\circ\pi_{1}+\iota_{2}\circ\pi_{2}=\id_{x_{1}\oplus x_{2}}$.

(4. $\Longrightarrow$ 5.) 
The conditions (a),(c) are implied by $\pi_{1}=\row{\id_{x_{1}}}{0_{x_{2},x_{1}}}$ and 
(b),(d) are implied by $\pi_{2}=\row{0_{x_{1},x_{2}}}{\id_{x_{2}}}$.
By the definition of coproducts, there is the unique arrow $k\colon\,x_{1}\oplus x_{2}\to x_{1}\oplus x_{2}$ satisfying 
$k\circ\iota_{1}=\iota_{1}$ and $k\circ\iota_{2}=\iota_{2}$.
Here, $k=\id_{x_{1}\oplus x_{2}}$ satisfies that and $k=\iota_{1}\circ\pi_{1}+\iota_{2}\circ\pi_{2}$ also satisfies as follows.
\[
k\circ\iota_{1}=(\iota_{1}\circ\pi_{1}+\iota_{2}\circ\pi_{2})\circ\iota_{1}=\iota_{1}\circ\pi_{1}\circ\iota_{1}+\iota_{2}\circ\pi_{2}\circ\iota_{1}=\iota_{1}\circ\id_{x_{1}}+\iota_{2}\circ 0_{x_{1},x_{2}}=\iota_{1}
\]
\[
k\circ\iota_{2}=(\iota_{1}\circ\pi_{1}+\iota_{2}\circ\pi_{2})\circ\iota_{2}=\iota_{1}\circ\pi_{1}\circ\iota_{2}+\iota_{2}\circ\pi_{2}\circ\iota_{2}=\iota_{1}\circ 0_{x_{2},x_{1}}+\iota_{2}\circ\id_{x_{2}}=\iota_{2}
\]
Therefore, we have the condition (e) $\iota_{1}\circ\pi_{1}+\iota_{2}\circ\pi_{2}=\id_{x_{1}\oplus x_{2}}$.

(5. $\Longrightarrow$ 1.) 
Assume that $\pi_{1},\pi_{2},\iota_{1},\iota_{2}$ satisfy the conditions (a)-(e) for any objects $x_{1},x_{2}$.
Then, for any object $y$, for any arrow $f_{1}\colon\,y\to x_{1}$, for any arrow $f_{2}\colon\,y\to x_{2}$,
the arrow $\iota_{1}\circ f_{1}+\iota_{2}\circ f_{2}$ satisfies the following equations.
\[
\pi_{1}\circ(\iota_{1}\circ f_{1}+\iota_{2}\circ f_{2})=\pi_{1}\circ\iota_{1}\circ f_{1}+\pi_{1}\circ\iota_{2}\circ f_{2}=\id_{x_{1}}\circ f_{1}+0_{x_{2},x_{1}}\circ f_{2}=f_{1}
\]
\[
\pi_{2}\circ(\iota_{1}\circ f_{1}+\iota_{2}\circ f_{2})=\pi_{2}\circ\iota_{1}\circ f_{1}+\pi_{2}\circ\iota_{2}\circ f_{2}=0_{x_{1},x_{2}}\circ f_{1}+\id_{x_{2}}\circ f_{2}=f_{2}
\]
Conversely, if $h\colon\,y\to x_{1}\oplus x_{2}$ satisfies $\pi_{1}\circ h=f_{1}$ and $\pi_{2}\circ h=f_{2}$, then $h=\iota_{1}\circ f_{1}+\iota_{2}\circ f_{2}$ as follows.
\[
h=\id_{x_{1}\oplus x_{2}}\circ h=(\iota_{1}\circ\pi_{1}+\iota_{2}\circ\pi_{2})\circ h=\iota_{1}\circ\pi_{1}\circ h+\iota_{2}\circ\pi_{2}\circ h=\iota_{1}\circ f_{1}+\iota_{2}\circ f_{2}
\]
Therefore, $x_{1}\oplus x_{2}$ is a binary product of $x_{1},x_{2}$ with projections $\pi_{1},\pi_{2}$
satisfying $\col{f_{1}}{f_{2}}=\iota_{1}\circ f_{1}+\iota_{2}\circ f_{2}$. 
Similarly, $x_{1}\oplus x_{2}$ is also a binary coproduct of $x_{1},x_{2}$ with injections $\iota_{1},\iota_{2}$ satisfying $\row{f_{1}}{f_{2}}=f_{1}\circ\pi_{1}+f_{2}\circ\pi_{2}$. 
Therefore, $x_{1}\oplus x_{2}$ is a biproduct of $x_{1},x_{2}$ and $(C,0,\oplus,\pi,\iota)$ is semiadditive.
Moreover, we have the following equation.
\[
\begin{array}{ll}
   & \nabla_{y}\circ(f\oplus g)\circ\Delta_{x}\\
= & \row{\id_{y}}{\id_{y}}\circ
 \left[
  \begin{array}{cc}
    f & 0_{x,y} \\ 
    0_{x,y} & g\\ 
  \end{array}
  \right]
\circ\col{\id_{x}}{\id_{x}}\\
= & (\pi_{1}+\pi_{2})\circ(\iota_{1}\circ f\circ\pi_{1}+\iota_{2}\circ g\circ\pi_{2})\circ(\iota_{1}+\iota_{2})\\
= & (\pi_{1}+\pi_{2})\circ(\iota_{1}\circ f\circ\pi_{1}\circ(\iota_{1}+\iota_{2})+\iota_{2}\circ g\circ\pi_{2}\circ(\iota_{1}+\iota_{2}))\\
= & (\pi_{1}+\pi_{2})\circ(\iota_{1}\circ f+\iota_{2}\circ g)\\
= & ((\pi_{1}+\pi_{2})\circ\iota_{1}\circ f+(\pi_{1}+\pi_{2})\circ\iota_{2}\circ g)\\
= & f+g\\
\end{array}
\]  
\end{proof}

\begin{definition}
A {\bf semiadditive $\cmon$-category} is defined to be a tuple $(C,+,0,\oplus,\pi,\iota)$
satisfying one (i.e., all) of equivalent conditions of Theorem~\ref{thm:semiadditive}.
\end{definition}

\section{Spectral decompositions of arrows}\label{sec:decom}

This section defines a generalized notion of the spectral decomposition in semiadditive $\cmon$-categories.
Our general notion of spectral decomposition includes many examples, including 
separation of disconnected directed graphs,
separation of connected multi-valued directed graphs,
separation of connected directed graphs, and
equitable partition of undirected connected graphs.

\begin{theorem}\label{thm:decom}
Let $(C,+,0,\oplus,\pi,\iota)$ be a semiadditive $\cmon$-category,
let $x_{1},x_{2}$ be objects of $C$ and let $\lambda_{1}\colon\,x_{1}\to x_{1},\lambda_{2}\colon\,x_{2}\to x_{2}$ be arrows of $C$.
For object $c$ and arrow $f\colon\,c\to c$, the following are equivalent.
\begin{enumerate}
\item $c$ is a product of $x_{1},x_{2}$ whose projections $\rho_{1}\colon\,c\to x_{1},\rho_{2}\colon\,c\to x_{2}$ satisfy $\rho_{1}\circ f=\lambda_{1}\circ \rho_{1}$ 
and $\rho_{2}\circ f=\lambda_{2}\circ \rho_{2}$.


\begin{figure}[h]
 \begin{center}
  \includegraphics[bb=0 0 427 243,width=0.3\columnwidth]{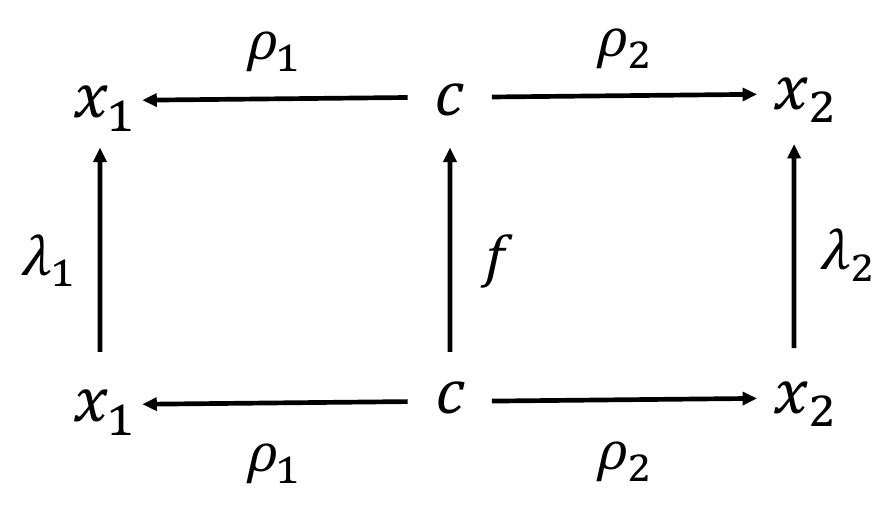}
 \end{center}
\end{figure}

\item $c$ is a coproduct of $x_{1},x_{2}$ whose injections $\kappa_{1}\colon\,x_{1}\to c,\kappa_{2}\colon\,x_{2}\to c$ satisfy $f\circ\kappa_{1}=\kappa_{1}\circ\lambda_{1}$ and $f\circ\kappa_{2}=\kappa_{2}\circ\lambda_{2}$.
\begin{figure}[h]
 \begin{center}
  \includegraphics[bb=0 0 427 239,width=0.3\columnwidth]{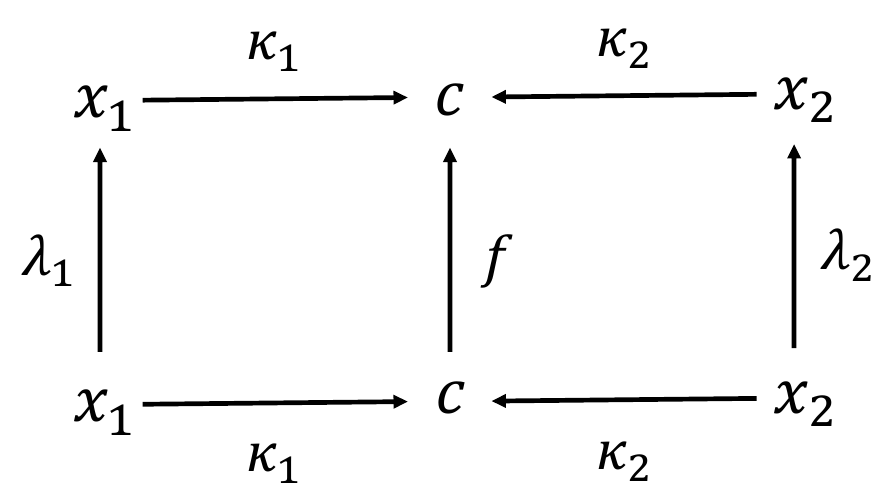}
 \end{center}
\end{figure}

\item There exist arrows 
$\rho_{1}\colon\,c\to x_{1},\rho_{2}\colon\,c\to x_{2},\kappa_{1}\colon\,x_{1}\to c,\kappa_{2}\colon\,x_{2}\to c$ such that
\begin{enumerate}
  \item $\rho_{1}\circ\kappa_{1}=\id_{x_{1}}$, $\rho_{2}\circ\kappa_{1}=0_{x_{1},x_{2}}$,
  \item $\rho_{1}\circ\kappa_{2}=0_{x_{2},x_{1}}$, $\rho_{2}\circ\kappa_{2}=\id_{x_{2}}$,
  \item $\kappa_{1}\circ\rho_{1}+_{c,c}\kappa_{2}\circ\rho_{2}=\id_{c}$, and
  \item $\kappa_{1}\circ\lambda_{1}\circ\rho_{1}+_{c,c}\kappa_{2}\circ\lambda_{2}\circ\rho_{2}=f$.
 \end{enumerate}
 \end{enumerate}
\end{theorem}
\begin{proof}
(1. $\Longrightarrow$ 3.) 
Since $c$ is a product of $x_{1},x_{2}$ whose projections $\rho_{1}\colon\,c\to x_{1},\rho_{2}\colon\,c\to x_{2}$,
there exists the unique $\kappa_{1}\colon\,x_{1}\to c$ satisfying the condition (a).
Similarly, there exists the unique $\kappa_{2}\colon\,x_{2}\to c$ satisfying the condition (b).
Both $h=\id_{c}$ and $h=\kappa_{1}\circ\rho_{1}+_{c,c}\kappa_{2}\circ\rho_{2}$ satisfy
$\rho_{1}\circ h=\rho_{1}$ and $\rho_{2}\circ h=\rho_{2}$ as follows.
\[
\rho_{1}\circ(\kappa_{1}\circ\rho_{1}+\kappa_{2}\circ\rho_{2}) = \rho_{1}\circ\kappa_{1}\circ\rho_{1}+\rho_{1}\circ\kappa_{2}\circ\rho_{2} = \id_{x_{1}}\circ\rho_{1}+0_{x_{2},x_{1}}\circ\rho_{2} = \rho_{1}+0_{c,x_{1}}= \rho_{1}
\]
\[
\rho_{2}\circ(\kappa_{1}\circ\rho_{1}+\kappa_{2}\circ\rho_{2}) = \rho_{2}\circ\kappa_{1}\circ\rho_{1}+\rho_{2}\circ\kappa_{2}\circ\rho_{2} = 0_{x_{1},x_{2}}\circ\rho_{1}+\id_{x_{2}}\circ\rho_{2}= 0_{c,x_{2}}+\rho_{2}= \rho_{2}
\]
Threfeore, the uniqueness of $h$ implies the condition (c) $\kappa_{1}\circ\rho_{1}+_{c,c}\kappa_{2}\circ\rho_{2}=\id_{c}$.
The condition (d) is satisfied as follows. 
\[
\begin{array}{ll}
   & \kappa_{1}\circ\lambda_{1}\circ\rho_{1}+\kappa_{2}\circ\lambda_{2}\circ\rho_{2}\\
= & \kappa_{1}\circ\rho_{1}\circ f+\kappa_{2}\circ\rho_{2}\circ f\\
= & (\kappa_{1}\circ\rho_{1}+\kappa_{2}\circ\rho_{2})\circ f\\
= & \id_{c}\circ f\\
= & f\\
\end{array}
\]

(3. $\Longrightarrow$ 1.) 
Assume that arrows $\rho_{1}\colon\,c\to x_{1},\rho_{2}\colon\,c\to x_{2},\kappa_{1}\colon\,x_{1}\to c,\kappa_{2}\colon\,x_{2}\to c$ satisfy conditions (a)-(d).
For any $g_{1}\colon\,c'\to x_{1}$ and $g_{2}\colon\,c'\to x_{2}$, $\kappa_{1}\circ g_{1}+\kappa_{2}\circ g_{2}$ is the unique arrow $h\colon\,c'\to c$ satisfying $\rho_{1}\circ h=g_{1}$ and $\rho_{2}\circ h=g_{2}$, as follows.
\[
\rho_{1}\circ(\kappa_{1}\circ g_{1}+\kappa_{2}\circ g_{2}) = \rho_{1}\circ\kappa_{1}\circ g_{1}+\rho_{1}\circ\kappa_{2}\circ g_{2} = \id_{x_{1}}\circ g_{1}+0_{x_{2},x_{1}}\circ g_{2} = g_{1}+0_{c',x_{1}}=g_{1}
\]
\[
\rho_{2}\circ(\kappa_{1}\circ g_{1}+\kappa_{2}\circ g_{2}) = \rho_{2}\circ\kappa_{1}\circ g_{1}+\rho_{2}\circ\kappa_{2}\circ g_{2} = 0_{x_{1},x_{2}}\circ g_{1}+\id_{x_{2}}\circ g_{2} = 0_{c',x_{2}}+g_{2}=g_{2}
\]
\[
h=\id_{c}\circ h=(\kappa_{1}\circ\rho_{1}+\kappa_{2}\circ\rho_{2})\circ h = \kappa_{1}\circ\rho_{1}\circ h+\kappa_{2}\circ\rho_{2}\circ h = \kappa_{1}\circ g_{1}+\kappa_{2}\circ g_{2}
\]
Threfore, $c$ is a product of $x_{1},x_{2}$ with projections $\rho_{1},\rho_{2}$.
The equation $\rho_{1}\circ f=\lambda_{1}\circ \rho_{1}$ and $\rho_{2}\circ f=\lambda_{2}\circ \rho_{2}$ holds as follows.
\[
\rho_{1}\circ f=\rho_{1}\circ(\kappa_{1}\circ\lambda_{1}\circ\rho_{1}+\kappa_{2}\circ\lambda_{2}\circ\rho_{2})=\rho_{1}\circ\kappa_{1}\circ\lambda_{1}\circ\rho_{1}+\rho_{1}\circ\kappa_{2}\circ\lambda_{2}\circ\rho_{2}=\lambda_{1}\circ\rho_{1}
\]
\[
\rho_{2}\circ f=\rho_{2}\circ(\kappa_{1}\circ\lambda_{1}\circ\rho_{1}+\kappa_{2}\circ\lambda_{2}\circ\rho_{2})=\rho_{2}\circ\kappa_{1}\circ\lambda_{1}\circ\rho_{1}+\rho_{2}\circ\kappa_{2}\circ\lambda_{2}\circ\rho_{2}=\lambda_{2}\circ\rho_{2}
\]

(2. $\Longrightarrow$ 3.) 
Since $c$ is a coproduct of $x_{1},x_{2}$ whose injections $\kappa_{1}\colon\,x_{1}\to c,\kappa_{2}\colon\,x_{2}\to c$,
there exist the unique $\rho_{1}\colon\,c\to x_{1}$ and the unique $\rho_{2}\colon\,c\to x_{2}$ satisfying the conditions (a),(b).
Both $h=\id_{c}$ and $h=\kappa_{1}\circ\rho_{1}+_{c,c}\kappa_{2}\circ\rho_{2}$ satisfy
$h\circ\kappa_{1}=\kappa_{1}$ and $h\circ\kappa_{2}=\kappa_{2}$ as follows.
\[
(\kappa_{1}\circ\rho_{1}+\kappa_{2}\circ\rho_{2})\circ\kappa_{1} = \kappa_{1}\circ\rho_{1}\circ\kappa_{1}+\kappa_{2}\circ\rho_{2}\circ\kappa_{1} = \kappa_{1}\circ\id_{x_{1}}+\kappa_{2}\circ 0_{x_{1},x_{2}}= \kappa_{1}+0_{x_{1},c}= \kappa_{1}
\]
\[
(\kappa_{1}\circ\rho_{1}+\kappa_{2}\circ\rho_{2})\circ\kappa_{2} = \kappa_{1}\circ\rho_{1}\circ\kappa_{2}+\kappa_{2}\circ\rho_{2}\circ\kappa_{2} = \kappa_{1}\circ 0_{x_{2},x_{1}}+\kappa_{2}\circ\id_{x_{2}}= 0_{x_{2},c}+\kappa_{2}= \kappa_{2}
\]
Threfeore, the uniqueness of $h$ implies the condition (c) $\kappa_{1}\circ\rho_{1}+_{c,c}\kappa_{2}\circ\rho_{2}=\id_{c}$.
The condition (d) is satisfied as follows. 
\[
\begin{array}{ll}
   & \kappa_{1}\circ\lambda_{1}\circ\rho_{1}+\kappa_{2}\circ\lambda_{2}\circ\rho_{2}\\
= & f\circ\kappa_{1}\circ\rho_{1}+f\circ\kappa_{2}\circ\rho_{2}\\
= & f\circ(\kappa_{1}\circ\rho_{1}+\kappa_{2}\circ\rho_{2})\\
= & f\circ\id_{c}\\
= & f\\
\end{array}
\]

(3. $\Longrightarrow$ 2.) 
Assume that arrows $\rho_{1}\colon\,c\to x_{1},\rho_{2}\colon\,c\to x_{2},\kappa_{1}\colon\,x_{1}\to c,\kappa_{2}\colon\,x_{2}\to c$ satisfy conditions (a)-(d).
For any $g_{1}\colon\,x_{1}\to c'$ and $g_{2}\colon\,x_{2}\to c'$, $g_{1}\circ\rho_{1}+g_{2}\circ\rho_{2}$ is the unique arrow $h\colon\,c\to c'$ satisfying $h\circ\kappa_{1}=g_{1}$ and $h\circ\kappa_{2}=g_{2}$, as follows.
\[
(g_{1}\circ\rho_{1}+g_{2}\circ\rho_{2})\circ\kappa_{1}=g_{1}\circ\rho_{1}\circ\kappa_{1}+g_{2}\circ\rho_{2}\circ\kappa_{1}=g_{1}\circ\id_{x_{1}}+g_{2}\circ 0_{x_{1},x_{2}}=g_{1}+0_{x_{1},c'}=g_{1}
\]
\[
(g_{1}\circ\rho_{1}+g_{2}\circ\rho_{2})\circ\kappa_{2}=g_{1}\circ\rho_{1}\circ\kappa_{2}+g_{2}\circ\rho_{2}\circ\kappa_{2}=g_{1}\circ 0_{x_{2},x_{1}}+g_{2}\circ\id_{x_{2}}=0_{x_{2},c'}+g_{2}=g_{2}
\]
\[
h=h\circ\id_{c}=h\circ(\kappa_{1}\circ\rho_{1}+\kappa_{2}\circ\rho_{2})= h\circ\kappa_{1}\circ\rho_{1}+h\circ\kappa_{2}\circ\rho_{2} = g_{1}\circ\rho_{1}+g_{2}\circ\rho_{2}
\]
Threfore, $c$ is a coproduct of $x_{1},x_{2}$ with injections $\kappa_{1},\kappa_{2}$.
The equation $f\circ\kappa_{1}=\kappa_{1}\circ\lambda_{1}$ and $f\circ\kappa_{2}=\kappa_{2}\circ\lambda_{2}$ holds as follows.
\[
f\circ\kappa_{1}=(\kappa_{1}\circ\lambda_{1}\circ\rho_{1}+\kappa_{2}\circ\lambda_{2}\circ\rho_{2})\circ\kappa_{1}=\kappa_{1}\circ\lambda_{1}\circ\rho_{1}\circ\kappa_{1}+\kappa_{2}\circ\lambda_{2}\circ\rho_{2}\circ\kappa_{1}=\kappa_{1}\circ\lambda_{1}
\]
\[
f\circ\kappa_{2}=(\kappa_{1}\circ\lambda_{1}\circ\rho_{1}+\kappa_{2}\circ\lambda_{2}\circ\rho_{2})\circ\kappa_{2}=\kappa_{1}\circ\lambda_{1}\circ\rho_{1}\circ\kappa_{2}+\kappa_{2}\circ\lambda_{2}\circ\rho_{2}\circ\kappa_{2}=\kappa_{2}\circ\lambda_{2}
\]
\end{proof}

\begin{definition}\label{def:decom}
In a semiadditive $\cmon$-category $(C,+,0,\oplus,\pi,\iota)$, a {\bf spectral decomposition} of $f\colon\,c\to c$ on $x_{1},x_{2}$ is 
defined to be a pair of $\lambda_{1}\colon\,x_{1}\to x_{1}$ and $\lambda_{2}\colon\,x_{2}\to x_{2}$ 
satisfying one (i.e., all) of equivalent conditions of Theorem~\ref{thm:decom}.
Moreover, $\kappa_{1}$ and $\kappa_{2}$ are called {\bf eigeninjections} which belong to $\lambda_{1}$ and $\lambda_{2}$, respectively.
\end{definition}

\begin{lemma}\label{lem:decom}
Let $(\lambda_{1},\lambda_{2})$ be a spectral decomposition of $f\colon\,c\to c$ on $x_{1},x_{2}$ with eigeninjections $\kappa_{1},\kappa_{2}$ in a semiadditive 
$\cmon$-category $(C,+,0,\oplus,\pi,\iota)$.
Moreover, let $(\lambda'_{1},\lambda'_{2})$ be a spectral decomposition of $f'\colon\,c\to c$ on $x_{1},x_{2}$ with the same eigeninjections $\kappa_{1},\kappa_{2}$.

Then, 
\begin{enumerate}
 \item $(\lambda'_{1}\circ\lambda_{1},\lambda'_{2}\circ\lambda_{2})$ is a spectral decomposition of $f'\circ f\colon\,c\to c$ on $x_{1},x_{2}$ with the same eigeninjections $\kappa_{1},\kappa_{2}$, and
 \item $(\lambda'_{1}+\lambda_{1},\lambda'_{2}+\lambda_{2})$ is a spectral decomposition of $f'+f\colon\,c\to c$ on $x_{1},x_{2}$ with the same eigeninjections $\kappa_{1},\kappa_{2}$.
\end{enumerate}
\end{lemma}
\begin{proof}
$\lambda_{1},\lambda'_{1}$, and $\kappa_{1}$ satisfy
$(f'\circ f)\circ\kappa_{1}=f'\circ\kappa_{1}\circ\lambda_{1}=\kappa_{1}\circ(\lambda'_{1}\circ\lambda_{1})$ and $(f'+f)\circ\kappa_{1}=f'\circ\kappa_{1}+f\circ\kappa_{1}=\kappa_{1}\circ\lambda'_{1}+\kappa_{1}\circ\lambda_{1}=\kappa_{1}\circ(\lambda'_{1}+\lambda_{1})$. 
$\lambda_{2},\lambda'_{2}$, and $\kappa_{2}$ also satisfy the similar equations.
\end{proof}

\begin{example}
In $\mat{\re}$ as semiadditive  category defined by Example 1, for each arrows 
$
h_{1},h_{2} : m \to n
$
, the $n \times m$ real matrices, a binary operator $h_{1} +_{m,n} h_{2}$ is defined by the matrix sum $ h_{1} + h_{2} $. 
{
Because 
\begin{align*}
h_{1} +_{m,n} h_{2}
&=
\nabla_{n}\circ(h_{1}\oplus h_{2})\circ\Delta_{m}
\\
&=
\row{\id_{n}}{\id_{n}}
\circ
 \left[
  \begin{array}{cc}
    h_{1} & 0_{m,n}\\ 
    0_{m,n} & h_{2} \\ 
  \end{array}
  \right]
\circ
\col{\id_{m}}{\id_{m}}
\\
&=
\rowmat{\id_{n}f_{n}}{\id_{n}g_{n}}
 \left(
  \begin{array}{cc}
    f^{-1}_{n}h_{1}f_{m} & f^{-1}_{n}0_{m,n}g_{m}\\ 
    g^{-1}_{n}0_{m,n}f_{m} & g^{-1}_{n}h_{2}g_{m} \\ 
  \end{array}
  \right)
\colmat{f^{-1}_{m}\id_{m}}{g^{-1}_{m}\id_{m}}
\\
&=
\rowmat{f_{n}}{g_{n}}\colmat{f^{-1}_{n}h_{1}}{g^{-1}_{n}h_{2}}
\\
&=
h_{1} + h_{2}, 
\end{align*}
where $f_{m}, g_{m}$ are $m$-dimensional regular matrices and $f_{n}, g_{n}$ are $n$-dimensional regular matrices defining projections and injections. 
}
Thus $\mat{\re}$ is a semiadditive $\cmon$-category. 

Let $f$ be an $n$-dimensional real symmetric matrix with real eigenvalues $\lambda_{1}, \lambda_{2}, \ldots ,\lambda_{n}$ and the corresponding eigenvectors $|v_{1}\rangle , |v_{2}\rangle ,\ldots |v_{n}\rangle $ with $f|v_{i}\rangle = \lambda_{i}|v_{i}\rangle $ for each $i=1,2,\ldots n$. Because $f$ is real symmetric matrix then $\{|v_{1}\rangle , |v_{2}\rangle ,\ldots |v_{n}\rangle \}$ can be an orthonormal basis of the $n$-dimensional real vector space $\re^{n}$. 
Therefore for each vector $|v\rangle \in \re^{n}$ we have a unique representation 
$
|v\rangle = \sum_{i=1}^{n}a_{i}|v_{i}\rangle .
$
In this setting, the action of the matrix $f$ is represented by 
$
f|v\rangle = \sum_{i=1}^{n}\lambda_{i}a_{i}|v_{i}\rangle .
$
Let 
$
\lambda_{i} : 1\to 1
$
be a multiplication operator 
$
\lambda_{i}(a_{i}|v_{i}\rangle )= \lambda_{i}a_{i}|v_{i}\rangle 
$ 
and 
$
\rho_{i} = |v_{i}\rangle \langle v_{i}|
$
where $\langle v_{i}|$ is the conjugate transpose of $|v_{i}\rangle $. Then we have
$
(\rho_{i}\circ f)|v\rangle 
=
\lambda_{i}a_{i}|v_{i}\rangle 
=
(\lambda_{i}\circ \rho_{i})|v\rangle 
$
for each $i=1,2,\ldots ,n$. Therefore the multiplication operators are satisfied with 1 of Theorem 3. 

Similarly, we can consider diagonalizable matrices cases. Let $f$ be an $n$-dimensional diagonalizable matrix with real eigenvalues $\lambda_{1}, \lambda_{2}, \ldots ,\lambda_{n}$ and the corresponding eigenvectors $|v_{1}\rangle , |v_{2}\rangle ,\ldots |v_{n}\rangle $ with $f|v_{i}\rangle = \lambda_{i}|v_{i}\rangle $ for each $i=1,2,\ldots n$. If we define a matrix $P = [ |v_{1}\rangle , |v_{2}\rangle ,\ldots |v_{n}\rangle ]$ then we can consider the inverse matrix 
$
P^{-1} 
=
\left[
\begin{array}{c}
\langle w_{1}| \\
\langle w_{2}| \\
\vdots \\
\langle w_{n}|\\
\end{array}
\right]
$ 
because $f$ is diagonalizable. 
Note that if we consider a projection 
$
\rho_{i} = |v_{i}\rangle \langle w_{i}|
$
then the corresponding injection should be 
$
\kappa_{i} = |v_{i}\rangle \langle w_{i}|
$
because $PP^{-1}=I_{n}$ induces 
$
\left(|v_{i}\rangle \langle w_{i}|\right)^{2} =|v_{i}\rangle \langle w_{i}|. 
$
Thus we have
$
(\lambda_{i}|v_{i}\rangle  \langle w_{i}|) |v\rangle 
=
\lambda_{i}a_{i}|v_{i}\rangle 
=
(\lambda_{i}\circ \rho_{i})|v\rangle 
=
(\kappa _{i} \circ \lambda_{i}\circ \rho_{i})|v\rangle 
$
for each $i=1,2,\ldots ,n$. This implies 
$
f
=
\kappa_{1}\circ\lambda_{1}\circ\rho_{1}+_{c,c}\kappa_{2}\circ\lambda_{2}\circ\rho_{2}+_{c,c}\cdots+_{c,c}\kappa_{n}\circ\lambda_{n}\circ\rho_{n}
=
\sum_{i=1}^{n}
\lambda_{i}|v_{i}\rangle  \langle w_{i}|,
$
the diagonalization of $f$. 

Lemma \ref{lem:decom} is corresponding to simultaneous diagonalization. Let $f, f^{\prime}$ be $n$-dimensional matrices with 
$
f
=
\sum_{i=1}^{n}
\lambda_{i}|v_{i}\rangle  \langle w_{i}|
$
and 
$
f^{\prime }
=
\sum_{i=1}^{n}
\lambda^{\prime }_{i}|v_{i}\rangle  \langle w_{i}|. 
$
Then we have 
$
f^{\prime }f
=
\sum_{i=1}^{n}
\lambda^{\prime }_{i}\lambda_{i}|v_{i}\rangle  \langle w_{i}|
$
and 
$
f^{\prime } + f
=
\sum_{i=1}^{n}
\left(\lambda^{\prime }_{i} + \lambda_{i}\right)
|v_{i}\rangle  \langle w_{i}|
$ 
by Lemma \ref{lem:decom}. These properties enable us constructing the spectral decomposition of matrix polynomials $\sum_{k=0}^{N}c_{k}f^{k}$ with $c_{0},\ldots ,c_{N}\in \re$ as 
\begin{align*}
\sum_{k=0}^{N}c_{k}f^{k}
=
\sum_{i=1}^{n}
\left(
\sum_{k=0}^{N}c_{k}\lambda _{i}^{k}
\right)
|v_{i}\rangle  \langle w_{i}|.
\end{align*}
The spectral decomposition of $\exp{f}=\lim_{n\to \infty}\Sigma_{k=0}^{n}\frac{1}{k!}f^{k}$ in $\mat{\re}$
is one of applications of Lemma~\ref{lem:decom}.
$\mat{\re}$ has scalar multiplication as arrows on any object.
When $(\lambda_{1},\lambda_{2})$ is a spectral decomposition of $f\colon\,c\to c$ on $x_{1},x_{2}$ with eigeninjections $\kappa_{1},\kappa_{2}$ in $\mat{\re}$,
the pair of scalar multiplications $\frac{1}{k!}\colon\,x_{1}\to x_{1}$ and $\frac{1}{k!}\colon\,x_{2}\to x_{2}$ is a
spectral decomposition of $\frac{1}{k!}\colon\,c\to c$.
By using Lemma~\ref{lem:decom} repeatedly, the pair of $\Sigma_{k=0}^{n}\frac{1}{k!}\lambda_{1}^{k}$ and 
 $\Sigma_{k=0}^{n}\frac{1}{k!}\lambda_{2}^{k}$ is a spectral decomposition of
 $\Sigma_{k=0}^{n}\frac{1}{k!}f^{k}$ for each non-negative integer $n$.
Since $\lim_{n\to \infty}$ is preserved by composition of arrows in $\mat{\re}$, 
$\exp{\lambda_{1}}$ and $\exp{\lambda_{2}}$ is a spectral decomposition of
$\exp{f}$.
However, limits are not preserved in any semiadditive $\cmon$-category.
General proof about preservation of limits needs
orders in semiadditive $\cmon$-categories
and needs to regard $\lim_{n\to \infty}g_{n}$ as the least upper bound
or the greatest lower bound of the arrow sequence $g_{0},g_{1},\ldots$ in the orders. 
\end{example}

\begin{example}
In $\mat{\posre}$ as a semiadditive $\cmon$-category defined by the same way as $\mat{\re}$, let $f$ be an $n$-dimensional block diagonal non-negative matrix 
$
f
=
\textrm{diag} (\lambda_{1}, \lambda_{2}, \ldots ,\lambda_{m})
$
where $\lambda_{i}$ be an $n_{i}$-dimensional non-negative matrix for each $i=1, 2, \ldots , m$. 
In this case, $\pi_{i} = \iota_{i}$ is nothing but the projection matrix onto the $n_{i}$-dimensional subspace of which $\lambda_{i}$ is acting for each $i=1,2,\ldots ,m$.  We have 
$
\iota _{i} \circ \lambda_{i}\circ \pi_{i}
=
\textrm{diag} (0_{n_{1},n_{1}}, \ldots ,0_{n_{i-1},n_{i-1}}, \lambda_{i}, 0_{n_{i+1},n_{i+1}}, \ldots ,0_{n_{m},n_{m}}).
$
Thus 
$
f
=
\iota_{1}\circ\lambda_{1}\circ\pi_{1}+_{c,c}\iota_{2}\circ\lambda_{2}\circ\pi_{2}+_{c,c}\cdots+_{c,c}\iota_{n}\circ\lambda_{m}\circ\pi_{m}
=
\textrm{diag} (\lambda_{1}, \lambda_{2}, \ldots ,\lambda_{m}).
$
\end{example}

\begin{example}
In $\mat{\complex}$ as a semiadditive $\cmon$-category defined by the same way as $\mat{\re}$, let $f$ be an $n$-dimensional complex matrix, there exists a regular matrix $P$ then we have so-called Jordan normal form
$
P^{-1}fP
=
\textrm{diag} (\lambda_{1}, \lambda_{2}, \ldots ,\lambda_{m}),
$
where $\lambda_{i}$ is known as the Jordan block, an $n_{i}$-dimensional matrix corresponding to the $i$-th eigenvalue. 
As same as Example 7, $\pi_{i} = \iota_{i}$ is nothing but the projection matrix onto the $n_{i}$-dimensional subspace of which $\lambda_{i}$ is acting for each $i=1,2,\ldots ,m$.  We have 
$
\iota _{i} \circ \lambda_{i}\circ \pi_{i}
=
\textrm{diag} (0_{n_{1},n_{1}}, \ldots ,0_{n_{i-1},n_{i-1}}, \lambda_{i}, 0_{n_{i+1},n_{i+1}}, \ldots ,0_{n_{m},n_{m}}).
$
Thus 
$
f
=
\iota_{1}\circ\lambda_{1}\circ\pi_{1}+_{c,c}\iota_{2}\circ\lambda_{2}\circ\pi_{2}+_{c,c}\cdots+_{c,c}\iota_{n}\circ\lambda_{m}\circ\pi_{m}
=
\textrm{diag} (\lambda_{1}, \lambda_{2}, \ldots ,\lambda_{m}).
$
 
\end{example}

\begin{example}
In $\rel$ as semiadditive  category defined by Example \ref{example-rel}, for each arrows 
$
h_1,\,h_2 : A \to B
$
, the binary relations, 
a binary operator $h_1 +_{A,B} h_2$ is defined by the union $h_1 \cup h_2$. 
Because 
$$
\begin{array}{rcl}
&&h_1 +_{A,B} h_2\\[5pt]
&=&\nabla_B \circ (h_1 \oplus h_2) \circ \Delta_A\\[5pt]
&=&\row{\id_B}{\id_B} \circ 
 \left[
  \begin{array}{cc}
    h_1 & 0_{A,B} \\ 
    0_{A,B} & h_2\\ 
  \end{array}
  \right]
\circ
\col{\id_A}{\id_A}\\[15pt]
&=&
(\pi_1^{B,B}\cup \pi_2^{B,B})\circ
(\iota_1^{B,B}\circ h_1\circ \pi_1^{A,A}\,\cup \, \iota_2^{B,B}\circ h_2\circ \pi_2^{A,A}) \circ
(\iota_1^{A,A}\cup \iota_2^{A,A})\\[10pt]
&=&
(\pi_1^{B,B}\circ \iota_1^{B,B}\circ h_1\circ \pi_1^{A,A}\,\cup \, 
\pi_2^{B,B}\circ \iota_2^{B,B}\circ h_2\circ \pi_2^{A,A}) \circ
(\iota_1^{A,A}\cup \iota_2^{A,A})\\[10pt]
&=&
(h_1\circ \pi_1^{A,A}\,\cup \, 
h_2\circ \pi_2^{A,A}) \circ
(\iota_1^{A,A}\cup \iota_2^{A,A})\\[10pt]
&=&
h_1\circ \pi_1^{A,A}\circ\iota_1^{A,A}
\,\cup \, 
h_2\circ \pi_2^{A,A}\circ\iota_2^{A,A}\\[10pt]
&=&
h_1\cup h_2
\end{array}
$$
Thus $\rel$ is a semiadditive $\cmon$-category. 

Let $C$ be the disjoint union of $A_1$ and $A_2$, and $f:C\to C$ be a homogeneous relation satisfying   
$f\subseteq (A_1\times A_1) \cup (A_2\times A_2)$. 
For $i\in \{1, 2\}$, we define the arrows 
$\rho_i:C\to A_i$, 
$\kappa_i:A_i\to C$, 
$\lambda_i:A_i\to A_i$, 
as follows:
$$
\begin{array}{l}
\rho_i=\{(a, a)\mid a\in A_i\}, \quad
\kappa_i=\{(a, a) \mid a\in A_i\}, \quad
\lambda_i=f|_{A_i}.
\end{array}
$$
where $\lambda_1$ and $\lambda_2$ are the restrictions of $f$ for each set respectively. 
The arrows $\rho_1$ and $\kappa_1$ are practically $\id_{A_1}$, and also $\rho_2$ and $\kappa_2$ are $\id_{A_2}$.

Then it holds that 
$\rho_i \circ f=\lambda_i\circ \rho_i$, 
${f \circ \kappa_i=\kappa_i\circ \lambda_i}$, 
and $\rho_i\circ \kappa_i=\id_{A_i}$ 
for each $i\in \{1, 2\}$. 
And also the followings hold: 
$\rho_2\circ \kappa_1=0_{A_1,A_2}$, 
$\rho_1\circ \kappa_2=0_{A_2,A_1}$, 
${\kappa_1\circ\rho_1 +\kappa_2\circ\rho_2=\id_C}$, 
$\kappa_1\circ\lambda_1 \circ\rho_1 +_{C,C} \kappa_2\circ\lambda_2 \circ\rho_2=f$. 
Thus it is shown that $(\lambda_1,\lambda_2)$ is a spectral decomposition of ${f:C\to C}$ on $A_1,\,A_2$. 

Let $(\lambda_1',\lambda_2')$ be a spectral decomposition of ${f':C\to C}$ on $A_1,\,A_2$, constructed as similar as $f$, that is, defined by $\lambda_1'=f'|_{A_1}$, $\lambda_2'=f'|_{A_2}$, $\rho_1,\,\rho_2,\,\kappa_1$, and $\kappa_2$.
Then $(\lambda'_{1}\circ\lambda_{1},\lambda'_{2}\circ\lambda_{2})$ is a spectral decomposition of $f'\circ f\colon\,C\to C$ on $A_1,A_2$ with the same eigeninjections $\kappa_{1},\kappa_{2}$,  
because for each $i\in \{1, 2\}$
$$\lambda'_{i}\circ\lambda_{i}=f'|_{A_i}\circ f|_{A_i}=(f'\circ f)|_{A_i}\,.$$
Moreover, $(\lambda'_{1}+\lambda_{1},\lambda'_{2}+\lambda_{2})$ is a spectral decomposition of $f'+f\colon\,C\to C$ on $A_1,A_2$ with the same eigeninjections $\kappa_{1},\kappa_{2}$, 
because 
$$\lambda'_{i}+\lambda_{i}=f'|_{A_i}\cup f|_{A_i}=(f'\cup f)|_{A_i}=(f'+ f)|_{A_i}$$
for each $i\in \{1, 2\}$.

\end{example}

\begin{example}\label{exam:4-spectral}
In $\rell{L}$ as semiadditive  category defined by Example \ref{example-Lrel}, for each arrows 
$h_1,h_2 : A \to B$, the $L$-relations, 
a binary operator $h_1 +_{A,B} h_2$ is defined by 
the join $h_1\sqcup h_2$ of the $L$-relations 
because of the similar calculation of $\rel$.
Thus $\rell{L}$ is a semiadditive $\cmon$-category. 

For the disjoint union $C$ of $A_1$ and $A_2$, let $f:C\to C$ be a homogeneous $L$-relation satisfying that $i\neq j$ implies $f(a,a')=\bot$ for each $a\in A_i$ and $a\in A_j$. 
For $i\in \{1, 2\}$, we define the arrows 
$\rho_i:C\to A_i$, 
$\kappa_i:A_i\to C$, 
$\lambda_i:A_i\to A_i$, 
by $\rho_i=\{(a, a)\mid a\in A_i\}$, 
$ \kappa_i=\{(a, a) \mid a\in A_i\}$, 
$\lambda_i=f|_{A_i}$. 
In this setting, $(\lambda_1,\lambda_2)$ is a spectral decomposition of ${f:C\to C}$ on $A_1,\,A_2$. This construction is obviously a generalization of previous example on $\rel$. 

The next example is a characteristic one on $\rell{L}$, which does not appear on $\rel$.
Consider the category $\rell{B_4}$ using the 4-elements Boolean algebra $B_4$. 
For a set $C=\{c_1, c_2\}$, let $f:C\to C$ satisfy $f(c_1, c_1)=a$, $f(c_2,c_2)=b$, $f(c_1,c_2)=f(c_2,c_1)=0$. 
For $A_1=\{\ast_1\}$ and $A_2=\{\ast_2\}$, we define the arrows 
$\rho_i:C\to A_i$, 
$\kappa_i:A_i\to C$, 
$\lambda_i:A_i\to A_i$, 
as follows:
$$
\begin{array}{l}
\rho_1(c_1,\ast_1)=\kappa_1(\ast_1,c_1)=a, \quad
\rho_2(c_1,\ast_2)=\kappa_2(\ast_2,c_1)=b,\quad\\
\rho_1(c_2,\ast_1)=\kappa_1(\ast_1,c_2)=b, \quad
\rho_2(c_2,\ast_2)=\kappa_2(\ast_2,c_2)=a,\quad\\
\lambda_1(\ast_1,\ast_1)=1, \quad\lambda_2(\ast_2,\ast_2)=0.
\end{array}
$$

These can be expressed in matrix form as follows:
$$f=
\begin{bmatrix}
a & 0\\
0 & b    
\end{bmatrix}
\quad
\kappa_1=
\begin{bmatrix}
a \\
b     
\end{bmatrix}
\quad
\kappa_2=
\begin{bmatrix}
b \\
a    
\end{bmatrix}
\quad
\rho_1=
\begin{bmatrix}
a & b    
\end{bmatrix}
\quad
\rho_2=
\begin{bmatrix}
b & a     
\end{bmatrix}
\quad
$$
Then it holds that $(\lambda_1,\lambda_2)$ is a spectral decomposition of ${f:C\to C}$ on $A_1,\,A_2$. 

For another example, we consider a set $C=\{c_1, c_2, c_3\}$, $A_1=\{x_1,x_2\}$, and $A_2=\{y_1\}$. A $L$-relation $f_1:C\to C$ satisfies $f_1(c_1, c_2)=a$ and  $f_1(c_2,c_3)=b$; otherwise $f_1(c,c')=0$,  
and the arrows 
$\rho_i:C\to A_i$, 
$\kappa_i:A_i\to C$, 
$\lambda_i:A_i\to A_i$
are defined as follows:
$$
\begin{array}{l}
f_1=
\begin{bmatrix}
0 & a & 0\\
0& 0 & b\\
0& 0 & 0\\
\end{bmatrix}
\quad
\rho_1=
\begin{bmatrix}
a & b & 0\\
0 & a & b\\    
\end{bmatrix}
\quad
\rho_2=
\begin{bmatrix}
b & 0 & a 
\end{bmatrix}
\quad
\kappa_1=
\begin{bmatrix}
a & 0\\
b & a\\
0 & b
\end{bmatrix}
\quad
\kappa_2=
\begin{bmatrix}
b \\
0 \\
a     
\end{bmatrix}
\quad
\\[20pt]
\lambda_1=
\begin{bmatrix}
0 & 1 \\
0 & 0
\end{bmatrix}
\quad
\lambda_2= 0
\end{array}
$$
Then it holds that $(\lambda_1,\lambda_2)$ is a spectral decomposition of ${f_1:C\to C}$ on $A_1,\,A_2$. 
The following diagram illustrates how $f_1$ is decomposed by $\lambda_1$ and $\lambda_2$ using a weighted directed graph.

Using the same arrows $\rho_i:C\to A_i$, $\kappa_i:A_i\to C$, \vspace{5pt}
we show another spectral decomposition $(\lambda'_1,\lambda'_2)$ of $L$-relation $f_2:C\to C$ defined by 
$
f_2=
\begin{bmatrix}
1& 0 & 0\\
0 & a & 0\\
0& 0 & b\\
\end{bmatrix}
$
, 
$
\lambda'_1=
\begin{bmatrix}
a & 0 \\
0 & 1
\end{bmatrix}
$, 
and $\lambda'_2= \begin{matrix} b \end{matrix}$.
The diagram in this setting is as follows:
\begin{figure}[h]
 \begin{center}
  \includegraphics[bb=0 0 473 362,width=0.4\columnwidth]{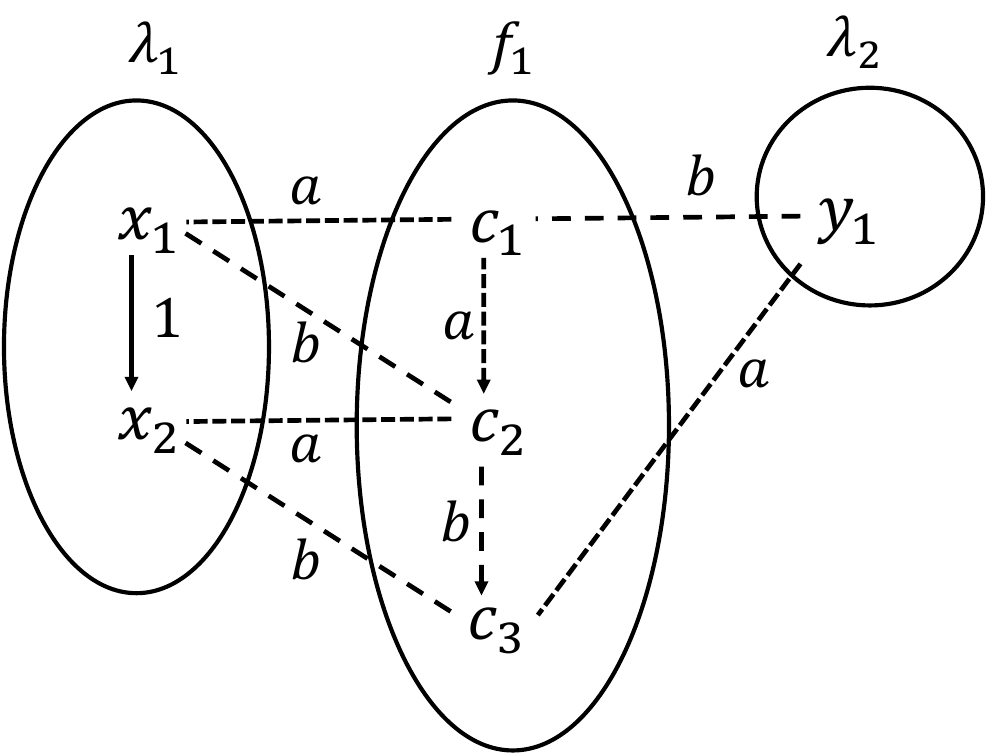}
  \includegraphics[bb=0 0 473 366,width=0.4\columnwidth]{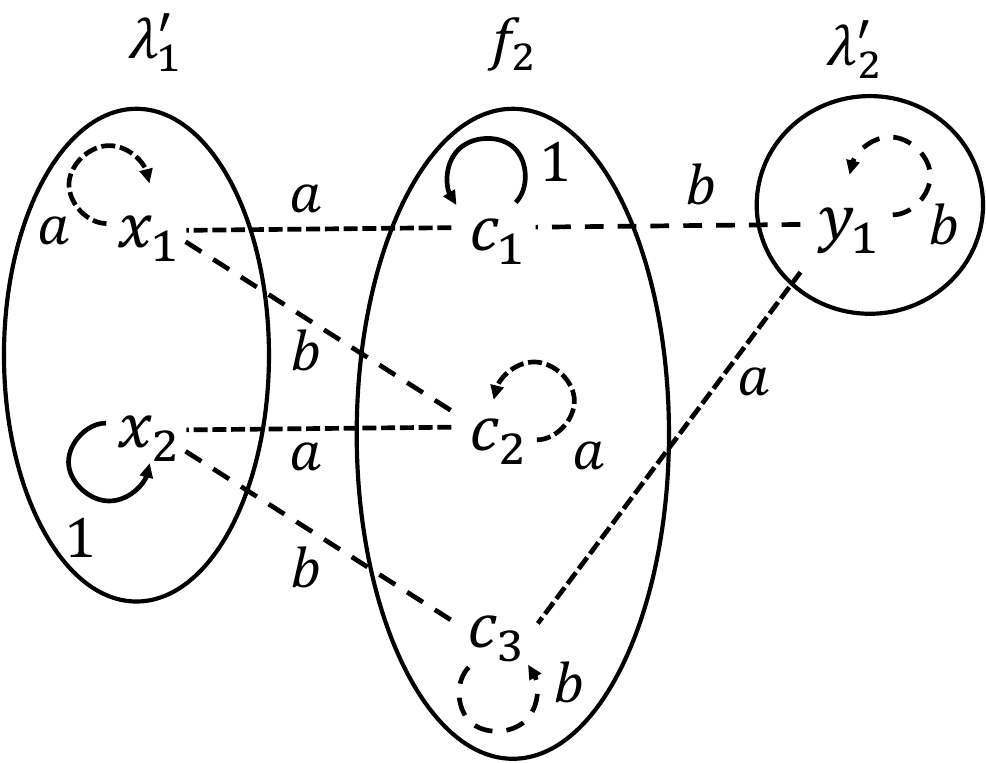}
 \end{center}
\end{figure}

By Lemma \ref{lem:decom}, $(\lambda'_1+\lambda_1, \lambda'_2+\lambda_2)$ is a spectral decomposition of $f_2+f_1:C\to C$ on $A_1$, $A_2$ as the following diagram. 
\begin{figure}[h]
 \begin{center}
  \includegraphics[bb=0 0 473 366,width=0.4\columnwidth]{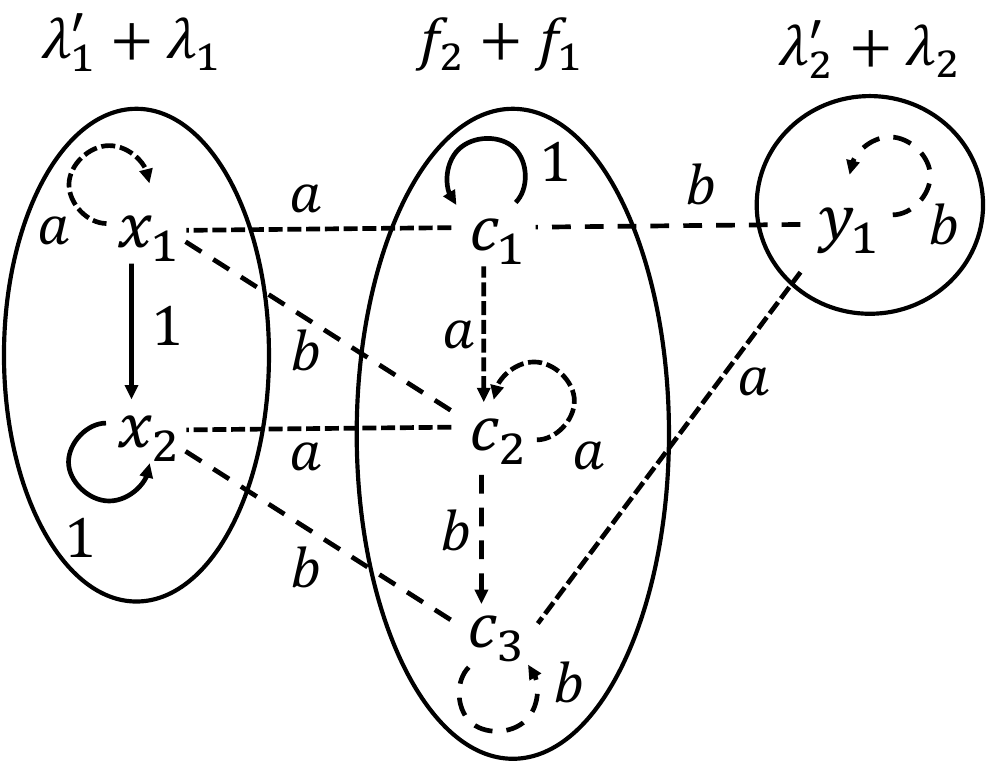}
 \end{center}
\end{figure}
\end{example}

\begin{example}
Let $G=(V(G),E(G))$ be an undirected connected graph with its vertex set $V(G)=\{1,2,\ldots ,n\}$ and edge set $E(G)$. We consider an $n$-dimensional real matrix $f$ with $\left( f \right)_{j,k}=\frac{1}{d_{j}}\  (j,k\in V(G))$ where $d_{j}=\sharp\{k\in V(G)\ |\ (j,k)\in E(G)\}$, the number of edges connected to the vertex $j$. Note that $d_{j}\geq 1$ for all $j\in V(G)$ because $G$ is connected. The matrix $f$ is nothing but the transition matrix of the simple random walk on $G$. Now we consider the following equitable partition  \cite{GodsilAGT} of $G$: 

A set of subgraphs $\left( G_{\bar{1}}, G_{\bar{2}}, \ldots , G_{\bar{N}} \right)$ is an equitable partition of the graph $G$ if it is satisfied with the following two conditions:
\begin{enumerate}
    \item
    $V(G) = \displaystyle\bigcup _{j=1}^{N}V(G_{\bar{j}}),\ V(G_{\bar{j}})\cap V(G_{\bar{k}}) = \emptyset \ (\bar{j}\neq \bar{k})$.
    \item 
    There exists a set of constants $\{d_{\bar{j},\bar{k}}\ |\ \bar{1}\leq \bar{j},\bar{k}\leq \bar{N}\}$ such that for all $\bar{1}\leq \bar{j},\bar{k}\leq \bar{N}$, $\sharp\{k\in V(G_{\bar{k}})\ |\ (j,k)\in E(G)\} = d_{\bar{j}, \bar{k}}$ for any $j\in V(G_{\bar{j}})$. 
\end{enumerate}
For an equitable partition $\left( G_{\bar{1}}, G_{\bar{2}}, \ldots , G_{\bar{N}} \right)$ of $G$, we put a projection
\begin{align*}
\rho_{1}
=
\begin{bmatrix}
    \bra{\bar{1}}\\
    \bra{\bar{2}}\\
    \vdots\\
    \bra{\bar{N}}
\end{bmatrix}
,\quad 
\bra{\bar{j}}
=
\frac{1}{n_{\bar{j}}}
\sum_{j\in V(G_{\bar{j}})}
\bra{j}
\ \ \text{with}\ \ 
n_{\bar{j}}=\sharp V(G_{\bar{j}})
\ \text{for}\ \ \bar{1}\leq \bar{j}\leq \bar{N}, 
\end{align*}
where $\{\bra{j}\}_{1\leq j\leq n}$ is the standard basis of $\mathbf{R}^{n}$ as row vectors. In this case, the corresponding injection $\kappa _{1}$ is given by 
\begin{align*}
\kappa_{1}
=
\begin{bmatrix}
    \widetilde{\ket{\bar{1}}}\ 
    \widetilde{\ket{\bar{2}}}\ 
    \ldots\
    \widetilde{\ket{\bar{N}}}
\end{bmatrix}
,\quad 
\widetilde{\ket{\bar{j}}}
=
\sum_{j\in V(G_{\bar{j}})}
\ket{j}
\ \text{for}\ \ \bar{1}\leq \bar{j}\leq \bar{N}, 
\end{align*}
where $\{\ket{j}\}_{1\leq j\leq n}$ is the standard basis of $\mathbf{R}^{n}$ as column vectors. Note that the action of the transition matrix $f$ to the vector $\bra{\bar{j}}$ is given by 
\begin{align*}
\bra{\bar{j}}f
=
\sum_{\bar{k}=\bar{1}}^{\bar{N}}
\left(
\frac{1}{n_{\bar{j}}}
\times 
\frac{d_{\bar{k},\bar{j}}}{d_{\bar{j}}}
\times 
n_{\bar{k}}
\right)
\bra{\bar{k}}
=
\sum_{\bar{k}=\bar{1}}^{\bar{N}}
\frac{d_{\bar{j},\bar{k}}}{d_{\bar{j}}}
\bra{\bar{k}}, 
\end{align*}
where 
$
d_{\bar{j}}
=
\sum_{\bar{k}=\bar{1}}^{\bar{N}}
d_{\bar{j},\bar{k}}.
$ 
Here we use the conservation law for the number of edges between two subgraphs $G_{\bar{j}}$ and $G_{\bar{k}}$, that is,  $n_{\bar{j}}d_{\bar{j},\bar{k}}=n_{\bar{k}}d_{\bar{k},\bar{j}}$. Thus we have the reduced $N\times N$ transition matrix $\lambda _{1}$ as 
\begin{align*}
\left(\lambda _{1}\right)_{\bar{j},\bar{k}}=\frac{d_{\bar{j},\bar{k}}}{d_{\bar{j}}}. 
\end{align*}

By Theorem \ref{thm:decom}, the transition matrix $f$ in $\mat{\re}$
is decomposed into two parts as
\begin{align*}
f
=
\kappa_{1}\lambda _{1}\rho_{1} + \left( f - \kappa_{1}\lambda _{1}\rho_{1}  \right)
=
\kappa_{1}\lambda _{1}\rho_{1} 
+
\kappa_{2}\lambda _{2}\rho_{2},
\end{align*}
if we define 
$
\kappa_{2}\lambda _{2}\rho_{2}
=
f - \kappa_{1}\lambda _{1}\rho_{1}. 
$
This shows that the first matrix $\kappa_{1}f\rho_{1}$ indicates the fundamental movement of the walkers, i.e., the walker walks on equitable partition $\left( G_{\bar{1}}, G_{\bar{2}}, \ldots , G_{\bar{N}} \right)$ according to the reduced transition matrix $\lambda _{1}$. The fluctuation of the total distribution of the walker is determined by the second matrix $\kappa_{2}\lambda _{2}\rho_{2}=f - \kappa_{1}\lambda _{1}\rho_{1}$. 
\end{example}

\begin{example}
The spectral decomposition in $\mat{\re}$ includes 
separation of connected directed graphs.

For example, $3$ is the biproduct of $2$ and $1$,
with the following projections $\rho_{1}:3\to 2$, $\rho_{2}:3\to 1$
and the following injections $\kappa_{1}:2\to 3$, and $\kappa_{2}:1\to 3$.
\[
\rho_{1}=
\begin{bmatrix}
1 & 0 & 0\\
0 & 1 & 1\\
\end{bmatrix},
\rho_{2}=
\begin{bmatrix}
0 & 0 & 1\\
\end{bmatrix},
\kappa_{1}=
\begin{bmatrix}
1 & 0\\
0 & 1\\
0 & 0\\
\end{bmatrix},
\kappa_{2}=
\begin{bmatrix}
0\\
-1\\
1\\
\end{bmatrix}
\]

The following arrow $f:3\to 3$ has a spectral decomposition $\lambda_{1}:2\to 2$, $\lambda_{2}:1\to 1$.
\[
f=
\begin{bmatrix}
0 & 1 & 1\\
1 & 0 & -1\\
0 & 0 & 1\\
\end{bmatrix},
\lambda_{1}=
\begin{bmatrix}
0 & 1\\
1 & 0\\
\end{bmatrix},
\lambda_{2}=
\begin{bmatrix}
1 \\
\end{bmatrix}
\]
This spectral decomposition can be regarded as
separation of connected graph $f$ on $\{c_{1},c_{2},c_{3}\}$ into $\{c_{1},c_{2}\}$ and $\{c_{3}\}$.
\begin{figure}[h]
 \begin{center}
  \includegraphics[bb=0 0 603 228,width=0.5\columnwidth]{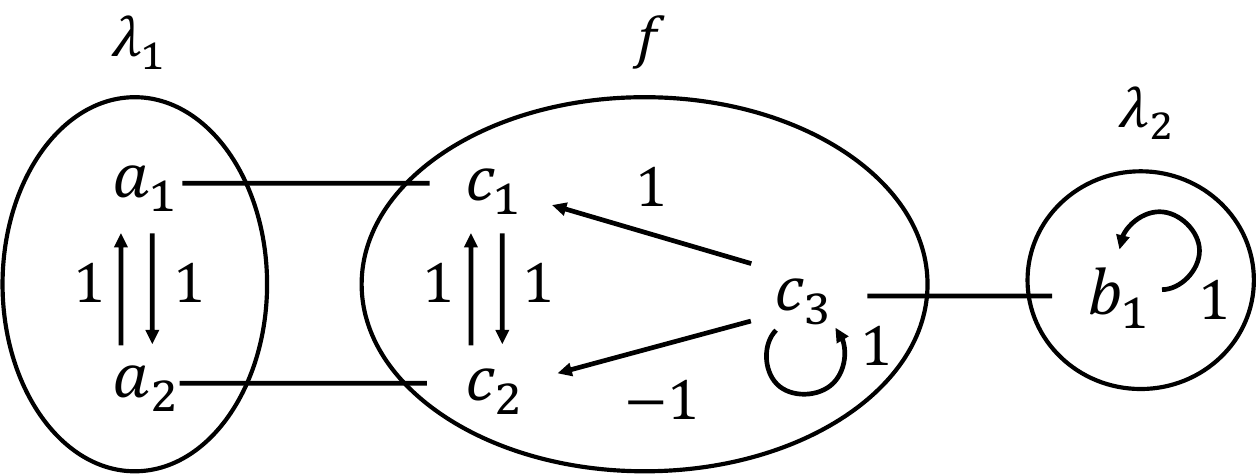}
 \end{center}
\end{figure}
\end{example}

\section{semiadditive $\cmon$-functors}\label{sec:functor}

\begin{theorem}\label{thm:functor}
Let $(C,+,0,\oplus,\pi,\iota),(C',+',0',\oplus',\pi',\iota')$ be 
semiadditive $\cmon$-categories.
For a functor $F\colon\,C\to C'$, the following are equivalent.
\begin{enumerate}
\item For any objects $x,y$ in $C$, the arrow part of $F$ is a monoid homomorphism from $(C(x,y),+_{x,y},0_{x,y})$ to 
$(C'(Fx,Fy),+'_{Fx,Fy},0'_{Fx,Fy})$, that is to say, 
\begin{itemize}
\item For any arrows $f\colon\,x\to y$, $g\colon\,x\to y$, $F(f+_{x,y}g)=Ff+'_{Fx,Fy}Fg$ and
\item $F0_{x,y}=0'_{Fx,Fy}$.
\end{itemize}
\item $F$ preserves the zero object and the binary biproducts, that is to say,
\begin{itemize}
\item for any zero object $c$ of $C$, $Fc$ is a zero object of $C'$ and
\item there exists a natural isomorphism
$\gamma\colon\,F(-\oplus-)\to (F-)\oplus'(F-)$ such that
for any objects $c_{1},c_{2}$ of $C$, 
$\pi'_{1}\circ\gamma_{c_{1},c_{2}}=F\pi_{1}$,
$\pi'_{2}\circ\gamma_{c_{1},c_{2}}=F\pi_{2}$,
$\iota'_{1}=\gamma_{c_{1},c_{2}}\circ F\iota_{1}$, and
$\iota'_{2}=\gamma_{c_{1},c_{2}}\circ F\iota_{2}$.
\end{itemize}
\end{enumerate}
\end{theorem}
\begin{proof}

(1. $\Longrightarrow$ 2.)
If $c$ is a zero object of $C$,  then $Fc$ is a zero object of $C'$, because $0'_{Fc,Fc}=F0_{c,c}=F\id_{c}=\id_{Fc}$.
For each pair of objects $c_{1},c_{2}$ of $C$,
we define $\gamma_{c_{1},c_{2}}\colon\,F(c_{1}\oplus c_{2})\to (F c_{1})\oplus'(F c_{2})$ by $\iota'_{1}\circ F\pi_{1}+'\iota'_{2}\circ F\pi_{2}$.
Then, it satisfies the following.
\[
\pi'_{1}\circ\gamma_{c_{1},c_{2}}=\pi'_{1}\circ(\iota'_{1}\circ F\pi_{1}+'\iota'_{2}\circ F\pi_{2})=\pi'_{1}\circ\iota'_{1}\circ F\pi_{1}+'\pi'_{1}\circ\iota'_{2}\circ F\pi_{2}=F\pi_{1}+'0'=F\pi_{1}
\]
\[
\begin{array}{ll}
  & \gamma_{c_{1},c_{2}}\circ F\iota_{1}\\
= & (\iota'_{1}\circ F\pi_{1}+'\iota'_{2}\circ F\pi_{2})\circ F\iota_{1}\\
= & \iota'_{1}\circ F\pi_{1}\circ F\iota_{1}+'\iota'_{2}\circ F\pi_{2}\circ F\iota_{1}\\
= & \iota'_{1}\circ F(\pi_{1}\circ\iota_{1})+'\iota'_{2}\circ F(\pi_{2}\circ \iota_{1})\\
= & \iota'_{1}\circ F(\id_{c_{1}})+'\iota'_{2}\circ F0\\
= & \iota'_{1}\circ \id_{Fc_{1}}+'\iota'_{2}\circ 0'\\
= & \iota'_{1}+'0'\\
= & \iota'_{1}\\
\end{array}
\]
Similarly, 
$\pi'_{2}\circ\gamma_{c_{1},c_{2}}=F\pi_{2}$ and 
$\gamma_{c_{1},c_{2}}\circ F\iota_{2}=\iota'_{2}$
are also satisfied.

$\gamma_{c_{1},c_{2}}$ is an isomorphism, because
$\gamma^{-1}_{c_{1},c_{2}}=F\iota_{1}\circ\pi'_{1}+'F\iota_{2}\circ\pi'_{2}$ is its inverse arrow.
$\gamma$ forms a natural transformation, 
because $\pi,\iota,\pi',\iota'$ are natural.

(2. $\Longrightarrow$ 1.)
Let $x,y$ be objects and let $c$ be a zero object in $C$.
Let $\sigma_{y}\colon\,c\to y$ be the unique arrow by the initial property of $c$.
Let $\tau_{x}\colon\,x\to c$ be the unique arrow by the terminal property of $c$.
Since $F$ preserves zero objects, $Fc$ is a zero object of $C'$, $F\sigma_{y}\colon\,Fc\to Fy$ is the unique arrow by the initial property of $Fc$,
and $F\tau_{x}\colon\,Fx\to Fc$ is the unique arrow by the terminal property of $Fc$.
Therefore, $F0_{x,y}=F(\sigma_{y}\circ\tau_{x})=F\sigma_{y}\circ F\tau_{x}=0_{Fx,Fy}$.

Let $f,g\colon\,x\to y$ be arrows in $C$. 
Since $F$ preserves binary biproducts with a natural isomorphism $\gamma$, we have
\[
\begin{array}{ll}
  & \pi'_{1}\circ\gamma_{Fx,Fx}\circ F\Delta_{x}\\
= & F\pi_{1}\circ F\Delta_{x}\\
= & F(\pi_{1}\circ \Delta_{x})\\
= & F\id_{x}\\
= & \id_{Fx}\\
\end{array}
\]
and $\pi'_{2}\circ\gamma_{Fx,Fx}\circ F\Delta_{x}=\id_{Fx}$ similarly.
Therefore, we have
$\gamma_{Fx,Fx}\circ F\Delta_{x} = \Delta_{Fx}$.

The equation
$\nabla_{Fy}= F\nabla_{y}\circ\gamma^{-1}_{Fy,Fy}$
also holds, because we have
\[
\begin{array}{ll}
  & F\nabla_{y}\circ\gamma^{-1}_{Fy,Fy}\circ\iota'_{1}\\
= & F\nabla_{y}\circ\gamma^{-1}_{Fy,Fy}\circ\gamma_{Fy,Fy}\circ F\iota_{1}\\
= & F\nabla_{y}\circ F\iota_{1}\\
= & F(\nabla_{y}\circ\iota_{1})\\
= & F\id_{y}\\
= & \id_{Fy}\\
\end{array}
\]
and $F\nabla_{y}\circ\gamma^{-1}_{Fy,Fy}\circ\iota'_{2}= \id_{Fy}$ similarly.

Therefore, $F$ preserves $+$ as follows.
\[
\begin{array}{ll}
  & F(f+_{x,y}g)\\
= & F(\nabla_{y}\circ(f\oplus g)\circ\Delta_{x})\\
= & F\nabla_{y}\circ F(f\oplus g)\circ F\Delta_{x}\\
= & F\nabla_{y}\circ\gamma^{-1}_{Fy,Fy}\circ\gamma_{Fy,Fy}\circ F(f\oplus g)\circ F\Delta_{x}\\
= & \nabla_{Fy}\circ\gamma_{Fy,Fy}\circ F(f\oplus g)\circ F\Delta_{x}\\
= & \nabla_{Fy}\circ(Ff\oplus'Fg)\circ\gamma_{Fx,Fx}\circ F\Delta_{x}\\
= & \nabla_{Fy}\circ(Ff\oplus'Fg)\circ\Delta_{Fx}\\
= & Ff+'_{Fx,Fy}Fg
\end{array}
\]
\end{proof}

\begin{definition}\label{def:functor}
For semiadditive $\cmon$-categories $(C,+,0,\oplus,\pi,\iota),(C',+',0',\oplus',\pi',\iota')$,  
a {\bf semiadditive $\cmon$-functor} $F\colon\,(C,+,0,\oplus,\pi,\iota)\to(C',+',0',\oplus',\pi',\iota')$ 
is defined to be a functor $F\colon\,C\to C'$ satisfying 
one (i.e., all) of equivalent conditions of Theorem~\ref{thm:functor}.
\end{definition}

\begin{definition}\label{def:preservation}
For semiadditive $\cmon$-categories $(C,+,0,\oplus,\pi,\iota),(C',+',0',\oplus',\pi',\iota')$,  
a functor $F\colon\,C\to C'$ is said to {\bf preserve all spectral decompositions}, if for any spectral decomposition $\lambda_{1},\lambda_{2}$ of $f\colon\,c\to c$ on $x_{1},x_{2}$ in $(C,+,0,\oplus,\pi,\iota)$, 
 the pair $F\lambda_{1},F\lambda_{2}$ 
 is also a spectral decomposition of $Ff\colon\,Fc\to Fc$ on $Fx_{1},Fx_{2}$
 in $(C',+',0',\oplus',\pi',\iota')$.
\end{definition}

\begin{theorem}
Every semiadditive $\cmon$-functor preserves all spectral decompositions.
\end{theorem}
\begin{proof}
Assume $F\colon\,(C,+,0,\oplus,\pi,\iota)\to(C',+',0',\oplus',\pi',\iota')$ 
is a semiadditive $\cmon$-functor
and a pair $\lambda_{1},\lambda_{2}$ is a spectral decomposition
 of $f\colon\,c\to c$ on $x_{1},x_{2}$ in $(C,+,0,\oplus,\pi,\iota)$. 
By Definition~\ref{def:decom}, 
there exist arrows 
$\rho_{1}\colon\,c\to x_{1},\rho_{2}\colon\,c\to x_{2},\kappa_{1}\colon\,x_{1}\to c,\kappa_{2}\colon\,x_{2}\to c$ such that
\begin{enumerate}
  \item $\rho_{1}\circ\kappa_{1}=\id_{x_{1}}$, $\rho_{2}\circ\kappa_{1}=0_{x_{1},x_{2}}$,
  \item $\rho_{1}\circ\kappa_{2}=0_{x_{2},x_{1}}$, $\rho_{2}\circ\kappa_{2}=\id_{x_{2}}$,
  \item $\kappa_{1}\circ\rho_{1}+_{c,c}\kappa_{2}\circ\rho_{2}=\id_{c}$, and
  \item $\kappa_{1}\circ\lambda_{1}\circ\rho_{1}+_{c,c}\kappa_{2}\circ\lambda_{2}\circ\rho_{2}=f$.
 \end{enumerate}
 
By Definition~\ref{def:functor}, the arrow part of $F$ is a monoid homomorphism.  
Therefore, 
 the pair $F\lambda_{1},F\lambda_{2}$ 
 is also a spectral decomposition of $Ff\colon\,Fc\to Fc$ on $Fx_{1},Fx_{2}$
 in $(C',+',0',\oplus',\pi',\iota')$, since $F\rho_{1}\colon\,Fc\to Fx_{1},F\rho_{2}\colon\,Fc\to Fx_{2},F\kappa_{1}\colon\,Fx_{1}\to Fc,F\kappa_{2}\colon\,Fx_{2}\to Fc$ satisfying the following equations.
\begin{enumerate}
  \item $F\rho_{1}\circ F\kappa_{1}=F(\rho_{1}\circ \kappa_{1})=F\id_{x_{1}}=\id_{Fx_{1}}$, \\
  $F\rho_{2}\circ F\kappa_{1}=F(\rho_{2}\circ\kappa_{1})=F0_{x_{1},x_{2}}=0_{Fx_{1},Fx_{2}}$,
  \item $F\rho_{1}\circ F\kappa_{2}=F(\rho_{1}\circ \kappa_{2})=F0_{x_{2},x_{1}}=0_{Fx_{2},Fx_{1}}$, \\
  $F\rho_{2}\circ F\kappa_{2}=F(\rho_{2}\circ\kappa_{2})=F\id_{x_{2}}=\id_{Fx_{2}}$,
  \item $F\kappa_{1}\circ F\rho_{1}+_{Fc,Fc}F\kappa_{2}\circ F\rho_{2}=F(\kappa_{1}\circ\rho_{1})+_{Fc,Fc}F(\kappa_{2}\circ\rho_{2})=F(\kappa_{1}\circ\rho_{1}+_{c,c}\kappa_{2}\circ\rho_{2})=F\id_{c}=\id_{Fc}$, and
  \item $F\kappa_{1}\circ F\lambda_{1}\circ F\rho_{1}+_{Fc,Fc}F\kappa_{2}\circ F\lambda_{2}\circ F\rho_{2}=F(\kappa_{1}\circ\lambda_{1}\circ\rho_{1})+_{Fc,Fc}F(\kappa_{2}\circ\lambda_{2}\circ\rho_{2})=F(\kappa_{1}\circ\lambda_{1}\circ\rho_{1}+_{c,c}\kappa_{2}\circ\lambda_{2}\circ\rho_{2})=Ff$.
\end{enumerate}
\end{proof}

\begin{example}
We give a semiadditive $\cmon$-functor for Example~\ref{exam:4-spectral} of $B_4=\{0,a,b,1\}$.
Let $F$ map a set $X$ to $X$ itself and $F$ map a 4-relation $f$
to a relation $F(f)=\{(x,y)\mid a\leq f(x,y)\}$.
That is to say, $F$ extracts only edges of the label $a$ or $1$ from a given 4-relation.

This $F$ is a semiadditive $\cmon$-functor from $\rell{B_4}$ to $\rel$.
This $F$ maps $\lambda'_{1}+\lambda_{1}$, $f_{2}+f_{1}$, and $\lambda'_{2}+\lambda_{2}$ of Example~\ref{exam:4-spectral}
to the following diagram.  
\begin{figure}[h]
 \begin{center}
  \includegraphics[bb=0 0 499 366,width=0.5\columnwidth]{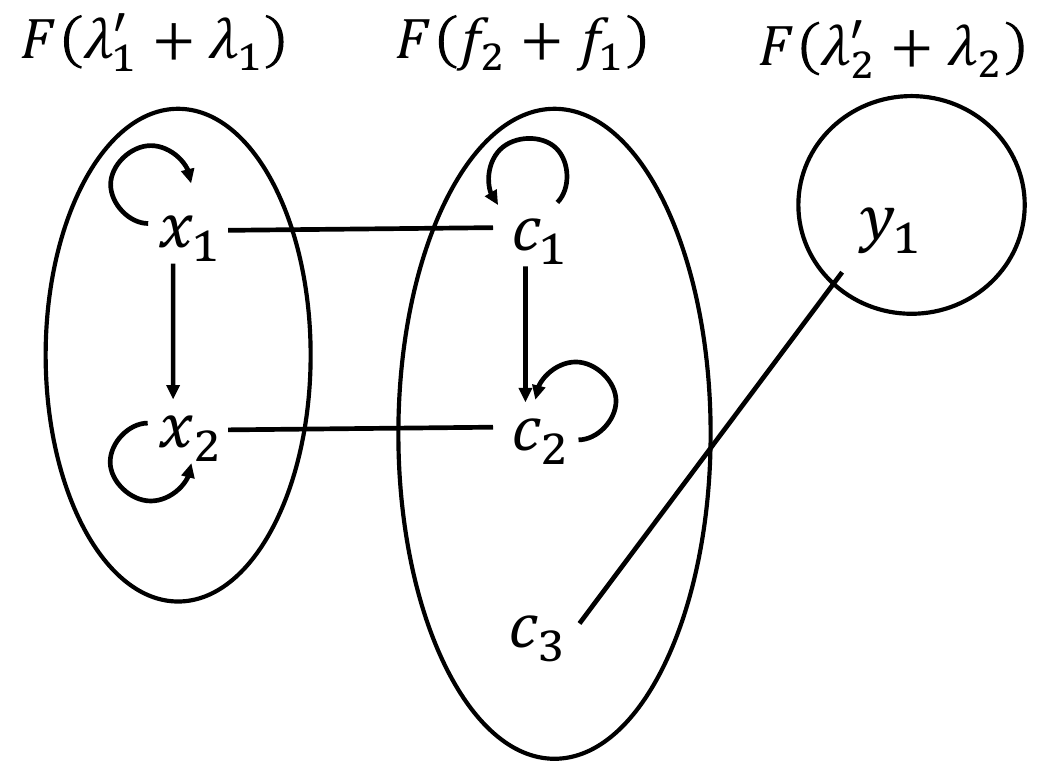}
 \end{center}
\end{figure}
\end{example}

\section{Conclusion and future work}\label{sec:conclusion}

In this paper, we gave several equivalent characterizations for a category with finite biproducts and the sum operation of arrows, and called categories satisfying these semiadditive $\cmon$-categories. 
This allowed us to give equivalent structures without directly confirming the existence of biproducts.
Moreover, we defined a generalized notion of the spectral decomposition in semiadditive $\cmon$-categories.
We also defined the notion of a semiadditive $\cmon$-functor that preserves the spectral decomposition of arrows.
Semiadditive $\cmon$-categories and semiadditive $\cmon$-functors
include many examples.

It is future work to construct projections, injections, and spectral decompositions for a given homogeneous arrow.

\bibliographystyle{splncs}
\bibliography{paper}

\begin{thebibliography}{1}

\bibitem{MacLane:Working}
MacLane, S.:
\newblock Categories for the Working Mathematician. Second edn.
\newblock Springer (1998)

\bibitem{Lack_2012}
Lack, S.:
\newblock Non-canonical isomorphisms.
\newblock Journal of Pure and Applied Algebra \textbf{216}(3) (March 2012)  593–597

\bibitem{Kelly1982}
Kelly, G.M.:
\newblock Basic concepts of enriched category theory.
\newblock Number~64 in London Mathematical Society Lecture Notes Series. Cambridge University Press (1982)

\bibitem{GodsilAGT}
Godsil, C., Royle, G.:
\newblock Algebraic Graph Theory.
\newblock Springer (2001)

\end{thebibliography}

\end{document}